\documentclass[11pt, reqno, a4paper]{amsart}
\usepackage[pagebackref, colorlinks = true, linkcolor = blue, urlcolor  = blue, citecolor = red]{hyperref}
\usepackage[table,xcdraw]{xcolor}
\usepackage{braket}
\usepackage{graphicx}
\usepackage{tikz}
\usepackage{amssymb,amsthm}
\usepackage[utf8]{inputenc} 
\usepackage[most]{tcolorbox}
\usepackage{svg}
\usepackage{bbm}
\usepackage{incgraph}
\usepackage[margin=1in]{geometry}
\usepackage{subfig}
\usepackage[capitalize]{cleveref}

\definecolor{mygreen}{RGB}{143,192,165}
\definecolor{myorange}{RGB}{255,176,32}

\newtheorem{innercustomthm}{Theorem}
\newenvironment{customthm}[1]
{\renewcommand\theinnercustomthm{#1}\innercustomthm}
{\endinnercustomthm}
\renewcommand{\epsilon}{\varepsilon}
\renewcommand{\phi}{\varphi}
\newcommand{\C}{\mathbb{C}}
\newcommand{\R}{\mathbb{R}}

\newcommand{\M}{\mathcal{M}}
\DeclareMathOperator{\Tr}{Tr}
\DeclareMathOperator{\id}{id}

\newcommand{\E}{\mathbb{E}}
 
\newcommand{\ketbra}[2]{|#1\rangle\langle#2|} 
\newcommand{\MP}{\mathrm{MP}}
\newcommand{\NC}{\operatorname{NC}}

\makeatletter 
\def\moverlay{\mathpalette\mov@rlay}
\def\mov@rlay#1#2{\leavevmode\vtop{
   \baselineskip\z@skip \lineskiplimit-\maxdimen
   \ialign{\hfil$\m@th#1##$\hfil\cr#2\crcr}}}
\newcommand{\charfusion}[3][\mathord]{
    #1{\ifx#1\mathop\vphantom{#2}\fi
        \mathpalette\mov@rlay{#2\cr#3}
      }
    \ifx#1\mathop\expandafter\displaylimits\fi}
\makeatother
\newcommand{\bigSsqcup}{\charfusion[\mathop]{\bigsqcup}{\innerS}}
\newcommand{\innerS}{\mathchoice
  {\mathsf S}
  {\scriptstyle \mathsf S}
  {\scriptscriptstyle \mathsf S}
  {\scriptscriptstyle \mathsf S}
}
\newcommand{\bigPsqcup}{\charfusion[\mathop]{\bigsqcup}{\innerP}}
\newcommand{\innerP}{\mathchoice
  {\mathsf P}
  {\scriptstyle \mathsf P}
  {\scriptscriptstyle \mathsf P}
  {\scriptscriptstyle \mathsf P}
}
\newtheorem{theorem}{Theorem}[section]
\newtheorem{definition}[theorem]{Definition}
\newtheorem{proposition}[theorem]{Proposition}
\newtheorem{corollary}[theorem]{Corollary}
\newtheorem{lemma}[theorem]{Lemma}
\newtheorem{remark}[theorem]{Remark}
\newtheorem{example}[theorem]{Example}

\DeclareMathAlphabet\mathbfcal{OMS}{cmsy}{b}{n}
\DeclareMathOperator{\rank}{rank}

\let\OLDthebibliography\thebibliography
\renewcommand\thebibliography[1]{
 \OLDthebibliography{#1}
\setlength{\parskip}{11pt}
\setlength{\itemsep}{2mm}
}

\begin{document}
\title{A Max-Flow approach to Random Tensor Networks}

\author{Khurshed Fitter}
\email{khurshed.fitter@epfl.ch}
\address{Ecole Polytechnique Federal de Lausane, Switzerland}
\author{Faedi Loulidi}
\email{faedi-loulidi@oist.jp}
\address{Okinawa Institute of Science Technology,Okinawa, Japan}

\author{Ion Nechita}
\email{nechita@irsamc.ups-tlse.fr}
\address{Laboratoire de Physique Th\'eorique, Universit\'e de Toulouse, CNRS, UPS, France}

\date{\today}

\begin{abstract}
We study the entanglement entropy of a \emph{random tensor network} (RTN) using tools from \emph{free probability theory}. Random tensor networks are simple toy models that help the understanding of the entanglement behavior of a boundary region in the ADS/CFT context. One can think of random tensor networks are specific probabilistic models for tensors having some particular geometry dictated by a graph (or network) structure. We first introduce our model of RTN, obtained by contracting maximally entangled states (corresponding to the edges of the graph) on the tensor product of Gaussian tensors (corresponding to the vertices of the graph). We study the entanglement spectrum of the resulting random spectrum along a given bipartition of the local Hilbert spaces. We provide the limiting eigenvalue distribution of the reduced density operator of the RTN state, in the limit of large local dimension. The limit value is described via a \emph{maximum flow} optimization problem in a new graph corresponding to the geometry of the RTN and the given bipartition. In the case of \emph{series-parallel} graphs, we provide an explicit formula for the limiting eigenvalue distribution using classical and free multiplicative convolutions. We discuss the physical implications of our results, allowing us to go beyond the semiclassical regime without any \emph{cut assumption}, specifically in terms of finite corrections to the average entanglement entropy of the RTN. 
\end{abstract}
\maketitle

\tableofcontents

\newpage

\section{Introduction}
The ADS/CFT
correspondence consists in describing a quantum theory (more precisely a \emph{conformal field theory})  as lying on the boundary of an \emph{anti de Sitter} space-time geometry \cite{maldacena1999large}. Many particular features of this correspondence remains mysterious, in particular the link with quantum information theory with \emph{entanglement}. It was shown in \cite{ryu2006holographic}, that for a fixed time slice, the entanglement behaviour of a given region of the boundary quantum theory is proportional to the minimal hypersurface bulk area homologous to the region of interest known as the Ryu-Takayanagi entanglement entropy. The Ryu-Tkayanagi formula shows in the context of ADS/CFT a crucial link between the entanglement behaviour of an intrinsic quantum theory and its link with the bulk gravitational field. This results open a new point on the understanding quantum gravity in the ADS/CFT framework from the perspective of entanglement and quantum information theory. We refer to \cite{chen2022quantum} and the reference therein for a complete introduction.

The difficulty of computing the entanglement properties of boundary quantum theories has led to the development of attractable simple models particularly the \emph{tensor network} and \emph{random tensor network} frameworks. Initially, the tensor network framework started as ``good'' models approximating ground states in condensed matter physics. 
In the context of condensed matter physics, tensor networks represent ground states of a class of gapped Hamiltonian \cite{cirac2021matrix}. Moreover, tensor networks have paved the way to understanding different physical properties such as the classification of topological phases of matter. 
We simply refer to \cite{cirac2021matrix} for an extensive review of all the different applications. Recently other extensions of tensor networks to random tensor networks for studying random matrix product states or projected entangled pairs of states were introduced in \cite{collins2013matrix,gonzalezguillen2018spectral,lancien2022correlation}. However, the random tensor network (or simply RTN) was initiated in \cite{hayden2016holographic} as toy models reproducing the key properties of the entanglement behaviour in the ADS/CFT context \cite{dong2021holographic,kudler2022negativity,penington2022replica,qi2022holevo,qi2018space,qi2017holographic,yang2016bidirectional}. Moreover the random tensor network framework appears in different active areas from condensed matter physics in the random quantum circuits and measurement framework \cite{levy2021entanglement,lopez2020mean,li2021statistical,medina2021entanglement,marolf2020probing,nahum2021measurement,vasseur2019entanglement,yang2022entanglement,you2018machine}.

In general, a random tensor network (or simply RTN) will consist of defining random quantum states from a given fixed graph structure, 
as we shall describe in the following lines. The main problem consists of computing the \emph{average entanglement} as $D\to\infty$, where $D$ plays the dimension of the Hilbert space of the model, behaviour of the state associated with a given fixed subregion of the graph. From Different results have been established allowing the understanding of the entanglement entropy of the RTN models as toy models mimicking the entanglement behaviour in quantum gravity. Different work has been explored in the literature where the entanglement entropy as $D\to\infty$ scales as the \emph{number of minimal cuts} needed to separate the region of interest from the rest of the graph times $\log D$ \cite{hayden2016holographic,cheng2022random}. Moreover one should mention that several directions have been explored to go beyond the toy model picture of the (random) tensor network \cite{apel2022holographic,bao2019beyond}. 

In this work, we will focus on a general random tensor network from a \emph{maximal flow} approach. The use of the maximal flow approach was already explored in \cite{kudler2022negativity} to compute the entanglement negativity and in \cite{freedman2017bit} to derive the Ryu-Takayanagi entanglement entropy in the continuum setting. As was described in the previous paragraph the model consists of defining a random quantum state from a given fixed graph structure. In our model, we shall consider a graph with edges (bulk edges) and half edges (boundary edges). The use of bulk and boundary edges will become clear from the definition of the model. We shall associate to each half edge a finite dimensional Hilbert space $\C^D$ and for each edge a Hilbert space $(\C^D)^{\otimes 2}$. The edges Hilbert space will generate a local Hilbert space associated to each vertex of the graph. In order to define an RTN, one should associate to each components of the graph quantum state generated as random. For that, we will generate for each vertex a random Gaussian state and we shall associate to each of the edges \emph{maximally entangled state}. A \emph{random tensor network} is defined by projecting all the maximally entangled states associated with the graph on the total random states generated in each vertex. The obtained random tensor network lies in the full boundary Hilbert space. The main goal of this work is to consider a sub-boundary region $A$ of the graph and evaluate the entanglement behaviour of the associated residual state as $D\to\infty$. The first computation is to estimate the moment computation of the state associated to the region $A$ as $D\to\infty$. With the help of the \emph{maximal flow}, that we will develop in this work in full details, we are able to estimate the moments \emph{without any cuts assumption} and shows that converges to the moment of a graph-dependent measure. We will show if the obtained \emph{partial order} is \emph{series-parallel}, and with the use of \emph{free probability theory}, we are able to explicitly construct the measure associated to the graph \emph{without any cut assumption}. Moreover, we will show the \emph{existence of higher order correction terms} of the entanglement entropy given with graph-dependent measure which can be explicitly given if the partial-order is series-parallel. We will show in different example how one can compute explicitly the measures associated to the initial graph in the case if the obtained partial order is series-parallel. The link between quantum information theory, free probability and random tensor network was already explored in \cite{cheng2024random} with the use of a general link state representing the effect of bulk matter field in the ADS/CFT context which allows to go beyond the semiclassical regime with correction terms of the entanglement structure. Moreover, the obtained results in \cite{cheng2024random} \emph{assumes} the existence of two disjoint minimal cuts separating the region $A$ from the rest. In this work, we only work with maximally entangled states in the bulk edge of the model. Without any cut assumption, we do obtain higher order correction terms which we may interpret as \emph{intrinsic fluctuations}. In the context of ADS/CFT those are intrinsic to the quantum spacetime nature of bulk gravitational field without any bulk matter field. 

This work is organised as follows. In Section \ref{sec: main results}, we will give a summary of our work by presenting all the main results. In Section \ref{sec: random tensor network} we will introduce our random tensor network framework. In Section \ref{sec: moment computation}, we will give the moment computation of a given state $\rho_A$ associated to a given suboudary region of the graph $A$. In Section \ref{sec: Optimisation with the MF}, with the help of the maximal flow approach, we will compute the asymptotic scaling of the moments and show the convergence to a measure given by a graph-dependent measure. In Section \ref{sec: entanglement and free probability theory}, we will introduce the notion of series-parallel partial order with the help of free probability we will show explicitly how one can construct a graph-dependent measure with free product convolution and classical measure product. In Section \ref{sec:examples}, we will give various examples of random tensor networks and show explicitly the associated obtained measure in the case of an obtained partial order is series-parallel. In Section \ref{sec:Results for normalized tensor network states}, we will give the main technical results, with the help of concentration inequalities we will show the obtained higher-order entanglement correction terms, without any cut assumption, are graph-dependent moreover the graph dependent measure can be explicitly constructed if the partial order is series-parallel.
    
\section{Main results}\label{sec: main results}
In this section, we will introduce the main definitions and main results obtained in this work. This work will consist of computing the entanglement entropy of a given random tensor network. We shall consider the most general framework of random tensor networks and study the entanglement structure of the random tensor, concerning a fixed bipartition of the total Hilbert space. By addressing the problem using a \emph{network flow approach}, we can compute the leading term, plus higher order correction terms of the entanglement entropy which are \emph{graph dependent}. The higher order correction terms plays a crucial role in different areas, particularly in the context of ADS/CFT \cite{lewkowycz2013generalized,hayden2016holographic,bao2019beyond,cheng2024random} as we shall comment after we give the main results. One can informally summarize the key results of this work as follows: 
\begin{tcolorbox}
In the limit of large local Hilbert space dimension $D$, the average R\'enyi entanglement entropy of a RTN $G$, across a given bipartition $(A|B)$, has:
\begin{itemize}
    \item a dominating term of the form $\operatorname{maxflow}(G_{A|B}) \cdot \log D$
    \item a finite correction term which is graph dependent.
\end{itemize}
In the case where $G_{A|B}$ is a \emph{series-parallel} graph, we can compute the distribution of the entanglement spectrum (and hence the finite entropy correction) as an iterative classical and free convolution of Mar\u{c}henko-Pastur distributions.
\end{tcolorbox}
A random tensor network has a corresponding random quantum state $\ket{\psi_G}$ that encodes the structure of a graph $G$. For that, we shall introduce the graph $G$ and some terminology. We refer to Section \ref{sec: random tensor network} for more technical details and definitions of the model.
Let $G=(V, E)$ be a connected undirected finite graph with (full) edges and half edges; the former encode the internal entanglement structure of the quantum state $\ket{\psi_G}$, while the latter represents the physical systems (Hilbert spaces) on which $\ket{\psi_G}$ lives. We shall denote by $E_b$ and $E_{\partial}$ the set of edges (bulk edges) and half edges (boundary edges) respectively. Formally the set of edges and half edges are respectively given by $E_b:=\{e_{x,y}\,|\,e_{x,y}=(x,y):\,x,y\in V\}$ and $E_{\partial}:=\{e_x=(x,\cdot): \,x\in V\}$ where $E:=E_b\sqcup E_{\partial}$. Then, the corresponding \emph{random tensor} $\ket{\psi_G}$ is defined as

\begin{equation*}
\ket{\psi_G}:=\Big\langle \bigotimes_{e\in E_b}\Omega_e \;  \Big | \; \bigotimes_{x\in V} g_x \Big\rangle \in  \big(\C^D\big)^{\otimes|E_{\partial}|},
    \end{equation*}
where $\ket{g_x}$ are random Gaussian states defined in the local Hilbert space of each vertex $x$. Moreover, for each (full) edge $e\in E_b$, we associate a maximally entangled state $\ket {\Omega_e} \in \C^D \otimes \C^D$ that is used to contract the internal degrees of freedom of the tensor network. For a representation of a random tensor network see Figure \ref{fig:tensor-network-with-marginal}, which we will treat in great detail for an illustration of our different main results of this work. We refer to Definition \ref{def: random tensor network} for more details.

As was mentioned earlier, this work aims to evaluate the \emph{entanglement entropy} of the random quantum state $\ket{\psi_G}$, along a bi-partition $A|B$ of the boundary edges $E_\partial = A \sqcup B$. We shall do so in the limit of large local Hilbert space dimension $D \to \infty$. To evaluate the entanglement entropy of the pure state $\ket{\psi_G}$, we shall compute its asymptotic \emph{entanglement spectrum} along the bi-partition $A|B$, that is the limiting spectrum of the density matrix $\rho_A=\Tr_B \ketbra{\psi_G}{\psi_G}$. From this spectral information, we can deduce the average R\'enyi entanglement and von Neumann entropies for the approximate normalised state $\tilde{\rho}_A$ respectively given by: 
\begin{equation*}
    \tilde{\rho}_A:=D^{-|E_{\partial}|}\rho_A\to \begin{cases}\lim_{D\to\infty}\E\,S_n(\tilde{\rho}_A)\quad\text{with}\quad S_n(\rho):=\frac{1}{1-n}\log\left(\Tr\rho^n\right), &\\
 \lim_{D\to\infty}\E\,S(\tilde{\rho}_A)\quad \text{with}\quad S(\rho):=-\Tr(\rho\log\rho).
    \end{cases}
\end{equation*}
Above, the expectation is taken with respect to the Gaussian distribution of the independent random tensors $\ket{g_x}$ present at each vertex of the graph. It will be clear from Section \ref{sec:Results for normalized tensor network states} the use of approximate normalised state instead of a ``true" normalised state $\tilde{\rho}_A:=\rho_A/\Tr\rho_A$. 

We first compute exactly \emph{the moments} of the random matrix $\rho_A$ and then we analyze the main contributing terms at large dimensions by relating the problem to a \emph{maximum flow} question in a related graph. By the use of the maximal flow and tools from \emph{free probability theory}, we will able to derive the leading and the fluctuating terms of the R\'enyi entropy and then deduce the behaviour of the von Neumann entanglement entropy.

\medskip

\noindent\textbf{Moment computation} We shall first consider the normalised state $\tilde{\rho}_A:=\rho_A/\Tr\rho_A$ and compute the moments. For the first step, we use the \emph{graphical Wick formula} from \cite{collins2011gaussianization} to find
\begin{equation}\label{equ: moments}
    \E \left[\Tr (\rho_A^n)\right]=\sum_{\alpha=(\alpha_{x})\in\mathcal{S}_n^{|V|}} D^{n|E_{\partial}|-H_G^{(n)}(\alpha)}, \;\forall\; n\in\mathbb N
\end{equation}
where $H_G^{(n)}(\alpha)$ can be understood as the Hamiltonian of a classical ``spin system'', where each spin variable takes a value from the permutation group $\mathcal{S}_n$:
$$ H_G^{(n)}(\alpha):=\sum_{(x,\cdot) \in A}|\gamma^{-1}_x\alpha_x|+\sum_{(x,\cdot)\in B}|\id^{-1}_x \alpha_x|+\sum_{(x,y)\in E_b}|\alpha_x^{-1}\alpha_y|.$$
Above, we associate to the region $B$, the identity permutation $\id_x \in \mathcal S_p$ (corresponding to taking the partial trace over $B$), and to the region $A$ the full-cycle permutation $\gamma_x = (n \; n-1 \; \cdots \; 2 \; 1)$ (corresponding to the trace of the $n$-th power of $\rho_A$). We refer to Proposition \ref{prop: average of moments} in Section \ref{sec: moment computation} for a more precise statement and proof. One should also mention that the contribution of the normalisation term of $\tilde\rho_A$ will be given by:
\begin{equation*}
    \E\left[(\Tr \rho_A)^n\right]=\sum_{\alpha=(\alpha_{x})\in\mathcal{S}_n^{|V|}} D^{{{n|E_{\partial}|-h_G^{(n)}(\alpha)}}}, \;\forall\; n\in\mathbb N
\end{equation*}
where 
\begin{equation*}
h_G^{(n)}(\alpha):=\sum_{(x,\cdot) \in E_{\partial}}|\id^{-1}_x\alpha_x|+\sum_{(x,y)\in E_b}|\alpha_x^{-1}\alpha_y|.
 \end{equation*}
Remark above that $h_G^{(n)}(\alpha)$ is simply $H_G^{(n)}(\alpha)$ with $A=\emptyset$. See Proposition \ref{prop: denominator contribution} for more details. Note that in the particular case $n=2$, the authors of \cite{hayden2016holographic} gave an exact mapping to the partition function of a classical Ising model. Notice the frustrated boundary conditions of the Hamiltonian above: vertices connected to the region $A$ prefer the configuration $\alpha_x = \gamma_x$, while vertices connected to the region $B$ prefer the low energy state $\alpha_x = \id_x$.

\noindent\textbf{Maximal flow.} The \emph{(max)-flow} approach will consist of identifying the leading terms from the moment formula above as $D\to \infty$. For that, we introduce a \emph{network} $G_{A|B}$, derived from the original graph $G$, by connecting all the half-edges in $A$ to an extra vertex $\gamma$ (sink) and all the half-edges in $B$ to $\id$ (source). In $G_{A|B}$, the vertices are valued in the permutation group $\mathcal S_n$ and all the half edges are connected either to the source $\id$ or to the sink $\gamma$. The flow approach will consist by looking at the different paths starting from the source $\id$ to the sink $\gamma$. The different paths in the flow approach will induce an ordering structure more precisely a \emph{poset structure} in the network $G_{A|B}$. Intuitively the maximal flow will consist of searching of the maximal number of such paths such that if on take  them off the source and the sink will be not anymore connected. More precisely, by \emph{Menger's theorem}, the maximum flow in this graph is equal to the number of edge-disjoint \emph{augmenting paths} that start from the source $\id$ and end in the sink $\gamma$. Figure \ref{fig:network} represents the different paths achieving the maximal flow in the network $G_{A|B}$ from the original graph $G$ as represented in Figure \ref{fig:tensor-network-with-marginal}. This procedure allows us to find a \emph{lower bound} to the Hamiltonian $H_{G_{A|B}}^{(n)}(\alpha)$ that can be attained by some choice of the variables $\alpha_x$. 

\begin{customthm}{A}
For all $n \geq 1$, we have
\begin{equation*}
    \min_{\alpha \in \mathcal S_n^{|V|}} H_{G_{A|B}}^{(n)}(\alpha)=(n-1)\operatorname{maxflow}(G_{A|B}),
\end{equation*}
\end{customthm}

where $H_{G_{A|B}}^{(n)}(\alpha)$ is the extended Hamiltonian in the network $G_{A|B}$. Once one takes out all the augmenting paths achieving the maximum flow in $G_{A|B}$, one is left with a \emph{clustered graph} $G_{A|B}^c$ that is obtained by clustering all the remaining connected components (see Figure \ref{fig:residual-network}). Importantly, it follows from the maximality of the flow that in this clustered graph, the cluster-vertices $[\id]$ and $[\gamma]$ are disjoint.  We refer to Proposition \ref{prop:min-H-max-flow} for more details and the proof of the result above.

As a direct consequence of the result above, one can deduce the moment convergence as $D\to\infty$, we refer to Theorem \ref{thm:limit-moments-general} for more details of the following result.  
\begin{customthm}{B}
    In the limit $D \to \infty$, we have, for all $n \geq 1$, 
    \begin{equation*}
        \lim_{D\to \infty}\E\frac{1}{D^{F(G_{A|B})}}\Big[\Tr\left(\left(D^{F(G_{A|B})-|E_{\partial}|}\,\rho_A\right)^n\right)\Big]=m_{n}
\end{equation*}
    where $m_n$ are the moments of a probability measure $\mu_{G_{A|B}}$ and $F(G_{A|B})=\operatorname{maxflow}(G_{A|B})$.
\end{customthm}
Moreover one can show the normalisation term converges to $1$ as shown in Corollary \ref{corr: average of normalisation}. 
The previous maximum flow computation gives the first order in the formula for the average entanglement entropy of random tensor network states: 
$$\mathbb{E}\left[S_n(\tilde \rho_A)\right] \approx  \operatorname{maxflow}(G_{A|B}) \cdot \log D \;\forall\; n \geq 1.$$
\noindent\textbf{Free probability theory and entanglement} Our main contribution in this work is to show that one can find the \emph{second order} (or the finite corrections)  of the R\'enyi and von Neumann entanglement entropy by carefully analyzing the set of augmenting paths achieving the maximum flow in the graph $G_{A|B}$. Once the different paths achieve the maximal flow in the graph $G_{A|B}$, after the clustering operation we obtain an \emph{partial order} $G_{A|B}^o$ where the vertices are the different permutation clusters formed from the clustered graph $G_{A|B}^c$. See Figure \ref{fig: order graph} of the obtained partial order from the original graph $G$ in Figure \ref{fig:tensor-network-with-marginal}. Our results are general, and they become explicit in the setting of the partial order $G^o_{A|B}$ is  \emph{series-parallel}. With the help of \emph{free probability theory}, we are able in this setting to deduce the second-order correction terms of each of the R\'enyi and von Neumann entropy.  
\begin{definition}
    A graph $G$ is called \emph{series-parallel} if it can be constructed recursively using the following two operations:
    \begin{itemize}
        \item Series concatenation: $G=H_1 \bigSsqcup H_2$ is obtained by identifying the sink of $H_1$ with the source of $H_2$.
        \item Parallel concatenation: $G=H_1\bigPsqcup H_2$ obtained by identifying the sources and the sinks of $H_1$ and $H_2$.
    \end{itemize}
\end{definition}
\begin{definition}
 To a series-parallel graph $G$ we associate a probability measure $\mu_G$, defined recursively as follows.
    \begin{itemize}
        \item To the trivial graph $G_{\text{triv}} = (\{s,t\}, \{\{s,t\}\})$, we associate the Dirac mass at $1$: $\mu_{G_{\text{triv}}} := \delta_1$.
        \item Series concatenation: $\mu_{G \bigSsqcup H} := \mu_G \boxtimes \MP \boxtimes \mu_H$.\footnote{$\mathrm{d}\MP:=\frac{1}{2\pi}\sqrt{4t^{-1}-1} \, \mathrm{d}t$ is the Mar\u{c}henko-Pastur distribution and $\boxtimes$ is the free convolution product. We refer to Appendix \ref{sec: appendix} for more details.}
        \item Parallel concatenation: $\mu_{G \bigPsqcup H} := \mu_G \times \mu_H.$
    \end{itemize}
\end{definition}

\begin{customthm}{C}
In the limit $D \to \infty$, the average R\'enyi entanglement entropy $\forall\; n \geq 1$ and von Neumann entropy of an approximate normalised state $\tilde \rho_A:=D^{-|E_{\partial}|}\rho_A$ behaves respectively as
    \begin{align*}
\mathbb{E}\left[S_n(\tilde \rho_A)\right] &=  \operatorname{maxflow}(G_{A|B}) \cdot \log D - \frac{1}{n-1}\log \int t^n \, \mathrm{d} \mu_{G_{A|B}}(t) + o(1)\\
\E \left[S(\tilde{\rho}_A)\right]&=\operatorname{maxflow}(G_{A|B}^o) \cdot \log D -\int t\,\log t\,\mathrm{d}\mu_{G_{A|B}}+o(1).
    \end{align*}
\end{customthm}
We refer to Corollary \ref{corr: average entropie scaling} for more details and the proof of the above statements. In particular if the obtained partial order $G_{A|B}^o$ is series-parallel the measure $\mu_{G_{A|B}}=\mu_{G_{A|B}^o}$ can be explicitly constructed, we refer to Theorem \ref{Th: moments and free pro} for more details. The use of the approximate normalised state instead of ``the'' normalised state $\tilde \rho_A:=\rho_A/\Tr\rho_A$ will be justified from the concentration result of $\Tr \rho_A$ in Subsection \ref{subsec: concentration}.

It was previously argued in \cite{hayden2016holographic,bao2019beyond} if one wants to encode the \emph{quantum fluctuations} one needs to use instead of a maximally entangled state a general ``link state" $\ket{\phi_e}$ defined by:
\begin{equation*}
    e\in E_b\to\ket{\phi_e}:=\sum_{i=1}^D\sqrt{\lambda_{e,i}}\ket{i_x,i_y}.
\end{equation*}
It was recently shown in \cite{cheng2024random} that the non-flat spectra of the link state under the existence \emph{assumption} of \emph{two non-disjoint cuts} that one obtains the quantum fluctuations beyond the semiclassical regime in ADS/CFT. The use of a generic link state in the context of ADS/CFT represents the bulk matter field contribution.
In this work with the \emph{maximal flow} approach, we were able to show the existence of quantum fluctuations without any \emph{minimal cut assumption} and with \emph{maximally entangled state} as link state. The obtained higher order correction terms in our context can be interpreted as the ``intrinsic" quantum fluctuations of spacetime geometry without any bulk matter field in the bulk represented by a general link state.

For example, in the case of the graph represented in Figure \ref{fig:tensor-network-with-marginal}, the resulting partial order $G_{A|B}^o$ is series-parallel  (see Figure \ref{fig: order graph}) where:
\begin{equation*}
        G_{A|B}^o=G_1 \bigSsqcup G_2 \bigSsqcup G_3\quad\text{with}\quad \mu_{G_{A|B}^o}= \mu_{G_1} \boxtimes \MP \boxtimes \mu_{G_2} \boxtimes \MP \boxtimes \mu_{G_3} =  \mu_{G_1} \boxtimes \MP^{\boxtimes 2},
\end{equation*}
as represented in Figure \ref{fig: G1G2G3} the graph $G_2$ and $G_3$ are trivial hence $\mu_{G_2} = \mu_{G_3} = \delta_1$. The graph $G_1$ can be factored as a parallel composition of two other graphs as represented in Figure \ref{fig: G4G5}:
\begin{equation*}
    G_1=G_5\bigPsqcup G_4\quad\text{with}\quad\mu_{G_1} = \mu_{G_4} \times \mu_{G_5}.
\end{equation*}
The graph $G_4$ as represented in Figure \ref{fig: G6G7G8} factorises as: 
\begin{equation*}
    G_4 = \left( G_6 \bigPsqcup G_7 \right) \bigSsqcup G_8\quad\text{with}\quad\mu_{G_4} = \left( \mu_{G_6} \times \mu_{G_7} \right) \boxtimes \MP \boxtimes \mu_{G_8}=\left( \MP \times \MP \right) \boxtimes \MP
\end{equation*}
where we have used the fact that $G_6$ and $G_7$ are series compositions of two trivial graphs, so $\mu_{G_6} = \mu_{G_7} = \MP$, while $\mu_{G_8} = \delta_1$.

Moreover the graph $G_5$ as represented in Figure \ref{fig G5} factorises as: 
\begin{equation*}
    G_5 = G_9 \bigSsqcup G_{10} \bigSsqcup \left( G_{11} \bigPsqcup G_{12} \right) \bigSsqcup G_{13}
    \end{equation*}
    with the associated measure
    \begin{equation*}
    \mu_{G_5} = \mu_{G_9} \boxtimes \MP \boxtimes \mu_{G_{10}} \boxtimes \MP \boxtimes \left( \mu_{G_{11}} \times \mu_{G_{12}} \right) \boxtimes \MP \boxtimes \mu_{G_{13}}=\MP^{\boxtimes 3} \boxtimes \left( \MP^{\boxtimes 2} \times \MP\right),
\end{equation*}
where we have used iteratively the series composition for $G_{11}$ and $G_{12}$ with their respective measure given by $\mu_{G_{11}} = \MP^{\boxtimes 2}$ and $\mu_{G_{12}} = \MP$.
 In the case of random tensor network represented in Figure \ref{fig:tensor-network-with-marginal} the partial order is series-parallel with the associated measure:
 \begin{equation*}
             G_{A|B}^o=G_1 \bigSsqcup G_2 \bigSsqcup G_3,
             \end{equation*}with
             \begin{equation*}
\mu_{G_{A|B}^o}=\left\{\left[ \MP^{\boxtimes 3} \boxtimes (\MP^{\boxtimes 2} \times \MP) \right] \times \left[ (\MP \times \MP) \boxtimes \MP \right] \right\} \boxtimes \MP^{\boxtimes 2},
 \end{equation*}
which is obtained by combining all the results stated above. If one considers the minimal cuts associated with the network $G_{A|B}$ (see Figure \ref{fig:network}) as represented in Figure \ref{fig:network-cuts} where we have considered four ways\footnote{We have only represented four cuts for simplicity. Remark in Figure \ref{fig:network-cuts} we have more than four minimal cuts which may share a common edge.} achieving the minimal cuts crossing common edges, therefore intersects. 

\bigskip

\section{Random tensor networks}\label{sec: random tensor network}
In this section, from a given graph with edges (bulk edges) and half edges (boundary edges), we will introduce random tensor network model. For that for each edge and half edge of the graph, we will associate a Hilbert space. The edge Hilbert space will induce a local Hilbert space for each vertex in the graph. We will associate to each of the vertices a random Gaussian state, and to each edge a maximally entangled state. The random tensor network is defined by projecting all the maximally entangled state associated to all edges of the graph over the vertex states given by the tensor product of all the random Gaussian vectors. This section aims to introduce the main definitions of the model and recall the different entanglement notions.

In Subsection \ref{subsec: random tensor network}, we shall introduce our random tensor network model. In Subsection \ref{subsec: entanglement of a subregion}, we recall the different entanglement notions and their properties. 

\subsection{Random tensor network}\label{subsec: random tensor network}
In the following, we shall give the construction of the random tensor network model. Let $G=(V, E)$ be a bulk connected undirected finite graph with edges and half edges. We shall denote by $E_b$ and $E_{\partial}$ the set of edges and half edges respectively. Formally the set of edges and half edges are defined as follows
\begin{align*}
     E_b&:=\{e_{x,y}\,|\,e_{x,y}=(x,y):\,x,y\in V\},\\
      E_{\partial}&:=\{e_x=(x,\cdot): \,x\in V\},\\
      E&:=E_b\sqcup E_{\partial}.
\end{align*}
For later discussion, the set of edges $E_b$ and half edges $E_{\partial}$ we shall call them the set of \emph{bulk} and \emph{boundary } edges. The bulk connectivity assumes that all the vertices in the bulk region of the graph are connected; this is the same notion as the ``connected network'' property from \cite[Definition 2]{hastings2017asymptotics}. We denote by $|E_b|$,  $|E_{\partial}|$ and $|E|=|E_b|+|E_{\partial}|$ the cardinality of the bulk, boundary and the total edge set. 

For each half-edge on a given vertex in the graph, we shall associate a Hilbert space $\C^{D}$, and for each bulk edge connecting two vertices, we associate $\C^{D}\otimes \C^{D}$ for finite $D$ known as the \emph{bond dimension}. We will define a random Gaussian to each vertex of the graph state that lies in the local Hilbert space associated to each vertex. Moreover, on each edge of the graph, we associate a maximally entangled state. The \emph{random tensor network} is a random quantum state constructed by projecting the total tensor product of the random Gaussian state for each vertex over all the maximally entangled state formed in bulk edges (see Definition \ref{def: random tensor network}).

Formally, for each part of the graph $G$ we shall associate to each part of the graph Hilbert spaces where: 
\begin{itemize}
    \item For each half-edge defined on a vertex $x$, we associate a finite-dimensional Hilbert space $\mathcal{H}_{e_x}$: 
\begin{align*}
    e_x\in E_{\partial}&\to \mathcal{H}_{e_x}:=\C^{D}\\
E_{\partial}&\to\mathcal{H}_{\partial}:=\bigotimes_{e_x\in E_{\partial}}\mathcal{H}_{e_x},
    \end{align*}

\item For each edges $e_{x,y}\in E_b$ we shall associate Hilbert space $\mathcal{H}_{e_{x,y}}$:
\begin{equation*}
    e_{x,y}\in E_b\to \mathcal{H}_{e_{x,y}}:=\C^{D}\otimes \C^{D}, 
\end{equation*}
where $\mathcal{H}_{e_{x,y}}$ denote the Hilbert space connecting the two vertices $x$ and $y$. 

\item For each vertex $x\in V$, we define the local vertex Hilbert space $\mathcal{H}_x$ where: 
\begin{align*}
    x\in V&\to \mathcal{H}_x:=\bigotimes_{E \ni e \to x} \mathcal H_e\\
    V&\to\mathcal{H}_V:=\bigotimes_{x\in V} \mathcal{H}_x = \bigotimes_{x \in V} \bigotimes_{E \ni e \to x} \mathcal H_e,
\end{align*}
where the Hilbert space $\mathcal{H}_x$ represents the local Hilbert space associated with a vertex $x$ defined as all the edges of Hilbert space that contribute locally. \end{itemize} 

Having defined the general Hilbert space structure associated with a generic graph $G$, in the following, we shall define quantum states in the graph $G$ which will allow us to introduce the random tensor network model. By construction let for each: 
\begin{itemize}
    \item Vertex $x$ a random quantum state $\ket{g_x}\in \mathcal{H}_x$ sampled from an \emph{i.i.d Gaussian distribution}:
    \begin{align*}
        x\in V&\to \ket{g_x}\in\mathcal{H}_x\\
        V&\to \bigotimes_{x\in V}\ket{g_x}\in \mathcal{H}_V        
    \end{align*}
\item Bulk edge $e_{x,y}$ a maximally entangled state $\ket{\Omega_e}$ given by: 
\begin{align*}
e_{x,y}\in E_b&\to\mathcal{H}_{e_{x,y}}\\
e_{x,y}&\to \ket{\Omega_e}:=\frac{1}{\sqrt{D}}\sum_{i=1}^{D}\ket{i_x,i_y},
\end{align*}
where we have used the notation $\ket{i_x}$ and $\ket{i_y}$ for the state associated to the vertex $x$ sharing an edge with $y$.   
\end{itemize}

\begin{definition}\label{def: random tensor network}
A \emph{random tensor network} $\ket{\psi_G}$ is defined as a projection of the vertex state over all the maximally entangled states $\ket{\Omega_e}$ for each $e_{x,y}$ in $E_b$ where:
\begin{equation}\label{eq: random tensor network}
    \ket{\psi_G}:=\Big\langle \bigotimes_{e\in E_b}\Omega_e \;  \Big | \; \bigotimes_{x\in V} g_x \Big\rangle \in  \big(\C^D\big)^{\otimes|E_{\partial}|}.
\end{equation}
\end{definition}
One should mention that the following example will be used in all other parts of this work as an illustration of the different results obtained in each section.
\begin{figure}[htpb]
\centering
\includegraphics{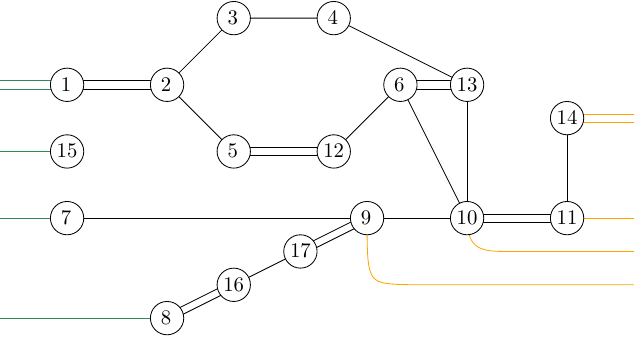}
\caption{A tensor network depicting a tensor from $(\mathbb C^D)^{\otimes 10}$ obtained by contracting 17 tensors. The 10 Hilbert space factors are partitioned into two subsets $\textcolor{mygreen}{B} \sqcup \textcolor{myorange}{A}$.}
\label{fig:tensor-network-with-marginal}
\end{figure}

\begin{example}
As an illustration of a random tensor, see Figure \ref{fig:tensor-network-with-marginal}, where the boundary region of $G$ are all the half edges $E_{\partial}:=\{e_{1},e_{15},e_{7},e_{8},e_{9},e_{10},e_{11},e_{14}\}$. We shall mention that in Figure \ref{fig:tensor-network-with-marginal}, the region $A$ are the half edges in the vertices $\{9,10,11,14\}$,i.e
$A:=\{e_{9},e_{10},e_{11},e_{14}\}$. The complementary region $B:=E_{\partial}\setminus A$ are half edges associated to the vertices $\{1,15,7,8\}$, where $B:=\{e_1,e_{15},e_7,e_8\}$. 
\end{example}

We shall also mention that our construction of the random tensor network, the edges, and the half edges generate the vertex Hilbert space $\mathcal{H}_x$. Other types of random tensor network models were already explored in the literature see \cite{hayden2016holographic,dong2021holographic,cheng2022random} and the reference therein. In the models mentioned previously, at first, they define the bulk and boundary vertices while in our work the focus is on the edges and the half edges which generates the local Hilbert space for each vertex, and the bulk states are given by a maximally entangled state. The first initial work in the random tensor network was in \cite{hayden2016holographic} where the aim was to compute the entanglement entropy of subregion of the random tensor network which is proportional as the bond dimension tends to infinity to the minimal cuts of the graph reproducing the famous Ryu-Takayanagi entanglement entropy \cite{ryu2006holographic} in a discrete version.

In a recent work \cite{cheng2022random}, the authors associate a state with a general ``link" state connecting two bulk vertices, therefore generalizing the previous models where they allowed the existence of two non-crossing minimal cuts. This result allows the authors to compute higher-order correction terms of the entanglement entropy.   
The main goal of this work, with the maximal flow approach without any minimal cut assumption, we will be able to derive the higher order correction terms with a maximally entangled state connecting the bulk vertices. 

\subsection{Entanglement }\label{subsec: entanglement of a subregion}
In the following, we shall recall different entanglement notions used in quantum information theory in particular von Neumann entropy and R\'enyi entropy.

The von Neumann entropie for a given normalised quantum state $\rho$ defined as 
\begin{equation}\label{VN entropy}
    S(\rho):=-\Tr\left(\rho\log\rho\right).
\end{equation}
In general, in physical systems with an exponential number of degrees of freedom it is in general difficult to compute it. There exists a generalisation where we do not need to diagonalise the density matrix $\rho$. This definition is due to Renyi which is known as the Renyi entropy defined as: 
\begin{equation}\label{renyi entropy}
S_n(\rho):=\frac{1}{1-n}\log\left(\Tr\rho^n\right),
\end{equation}
where it is well known that as $n\to 1$ the Renyi entropy converges to von Neumann entropy. The definitions given above are for normalised quantum states, if the state is not normalised one should normalise it first and then compute the entropy.\\

Now, we mention a bit about a subtlety regarding the upper bound on the rank of the reduced density matrix induced by the minimal cut. A minimal cut consists of finding the minimal number of edges in a graph that need to be removed to fully separate to a given fixed region of the graph. Although it is trivial to see that the rank of the reduced density matrix \(\rho_A\) is upper bounded by the local dimension \(D\) raised to the number of edges in the set \(A\), that is, \(\rank(\rho_A) \leq D^{\vert A \vert}\). However, there exists a subtlety. The rank of the reduced density matrix is, in fact, upper bounded by the minimum number of connecting edges or the bottleneck (min-cut) and not the number of edges:
\begin{equation}\label{eq:upper-bound-rank}
    \rank(\rho_A)\leq D^{F_A},
\end{equation}
where, \(F_A\) is the min-cut or the number of edges in the ``bottleneck".\\

Now, we demonstrate this more clearly using an example. Consider a state \(\ket{\psi_G}\), which we can use to construct \(\rho_A\) as shown below.
\begin{figure}[!h]
    \centering
    \includegraphics[width=0.9\textwidth]{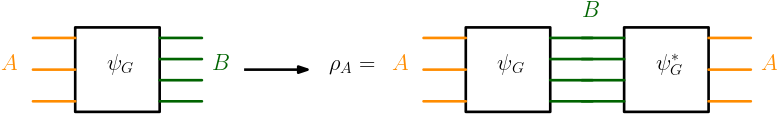}
    \caption{Constructing \(\rho_A\) from \(\ket{\psi_G}\) by tracing out the edges corresponding to the set \(B\).}
    \label{fig:enter-label}
\end{figure}
Now, consider the internal structure of \(\ket{\psi_G}\), where we divide the graph into two subgraphs denoted by \(L\) and \(R\), connected by the ``bottleneck" which is the set of all edges which when removed would disconnect the boundary sets \(A\) and \(B\).
\begin{figure}[!h]
    \centering
    \includegraphics[width=0.9\textwidth]{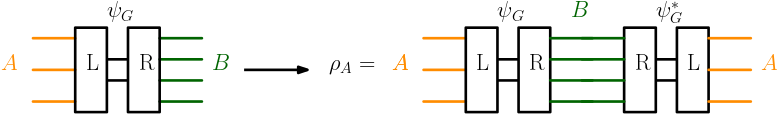}
    \caption{}
    \label{fig:enter-label2}
\end{figure}
Now, it is clear that \(\rank(\rho_A) \leq D^{F_A}\), where, in this case \(F_A = 2\), and consequently,
\begin{equation}
    S(\rho_A) \leq F_A \log D.
\end{equation}
Having established the natural intuition for the role of the min-cut (\(F_A\)) in upper-bounding the entropy, we now move on to establish our (maximal) flow approach for the random tensor network in the following sections.

\section{Moment computation}\label{sec: moment computation}

From a given random tensor network, we want to understand the behaviour of entanglement of a given subregion of the tensor network with the rest. 
For that we shall adress at first the moment computation of quantum state $\tilde {\rho}_A$ for a given subregion $A\subseteq E_{\partial}$. This first computation will allows us in the following sections to analyse the Renyi and the von Neumann entropy.  
\bigskip

Let $A\subseteq E_{\partial}$ be a sub-boundary region of the graph $G$. We shall denote by $B:=E_{\partial}\setminus A$ the complementary region of $A$. Let $\mathcal{H}_A:=\bigotimes_{e_x\in A}\mathcal{H}_{e_x}$ and $\mathcal{H}_{B}:=\bigotimes_{e_x\in B}\mathcal{H}_{e_x}$ respectively the Hilbert space associated to the boundary regions $A$ and $B$.

In this work, we will be interested in computing the \emph{average entanglement entropy} at large bond dimension: 
\begin{equation}\label{eq: large bond dimension average entanglement entropie}
    \tilde{\rho}_A:=\frac{\rho_A}{\Tr \rho_A}\to \begin{cases}\lim_{D\to\infty}\E\,S_n(\tilde{\rho}_A)&\\
 \lim_{D\to\infty}\E\,S(\tilde{\rho}_A),
    \end{cases}
\end{equation}
where $\tilde{\rho}_A$ is the normalised quantum state obtained by tracing out the region $B$, i.e $\rho_A=\Tr_B\ketbra{\psi_G}{\psi_G}$ where the partial trace over the Hilbert space $\mathcal{H}_B$. In the expression above, the average is over all the random Gaussian states. 

The first computation that we shall adress here is the moment computation as described in the following proposition. This will allow us later, as analysed in detail in the following sections, to compute the average entanglement entropy (R\'enyi and von Neumann entropy) as $D\to \infty$. The result above has been previously obtained in a very similar setting by Hastings \cite[Theorem 3, $E_{ind}$ ensemble]{hastings2017asymptotics}.

\begin{proposition}\label{prop: average of moments}
For any $A\subseteq E_{\partial}$, we have 

\begin{equation}\label{eq: expectation of moments}
    \E \left[\Tr (\rho_A^n)\right]=\sum_{\alpha=(\alpha_{x})\in\mathcal{S}_n^{|V|}} D^{{{n|E| - n|E_b|}}-H_G^{(n)}(\alpha)}, \;\forall\; n\in\mathbb N
\end{equation}
where $H_G^{(n)}(\alpha)$ can be understood as the Hamiltonian of a classical ``spin system'', where each spin variable takes a value from the permutation group $\mathcal{S}_n$:
\begin{equation}\label{eq: Hamiltonian}
 H_G^{(n)}(\alpha):=\sum_{(x,\cdot) \in A}|\gamma_x^{-1}\alpha_x|+\sum_{(x,\cdot)\in B}|\id_x^{-1} \alpha_x|+\sum_{(x,y)\in E_b}|\alpha_x^{-1}\alpha_y|.
 \end{equation}
\end{proposition}

Before giving the proof of the proposition above, we shall recall some properties of the permutation group $\mathcal{S}_n$ and fix some notations. We denote by $\gamma_x$ the \emph{total cycle} in the permutation group $\mathcal{S}_n$ evaluated in $(x,\cdot)\in A$ 
\begin{equation*}
    \forall (x,\cdot)\in A,\qquad\gamma_x=(n\ldots 1). 
\end{equation*}
We recall that one can define a notion of distance in $\mathcal{S}_n$ known as the \emph{Cayley distance} given by 
\begin{align*}
    \mathcal{S}_n\times\mathcal{S}_n&\to \R^+\\
    d:(\alpha_i,\alpha_j)&\to d(\alpha_i,\alpha_j):=n-\#(\alpha_i^{-1}\alpha_j),
\end{align*}
where $\#(\alpha)$ stands for the number of cycles in $\alpha$. The distance $d(\alpha_i,\alpha_j)$ gives the minimum number of transpositions to turn $\alpha_i$ to $\alpha_j$.
In general the distance in $\mathcal{S}_n$ satifies the triangle inequality where: 
\begin{equation*}
    d(\alpha_i,\alpha_j)\leq d(\alpha_i,\sigma)+d(\sigma,\alpha_j). 
\end{equation*}
In particular, we say that $\sigma$ is a \emph{geodesic} between $\alpha_i$ and $\alpha_j$ in $\mathcal S_n$ if $d(\alpha_i,\alpha_j)=d(\alpha_i,\sigma)+d(\sigma,\alpha_j)$.
We shall adopt the following notation for the distance instead of $d(\cdot,\cdot)$ where
\begin{equation*}
(\alpha_i,\alpha_j)\in\mathcal{S}_n\times\mathcal S_n,\quad d(\alpha_i,\alpha_j)=|\alpha_i^{-1}\alpha_j|.
\end{equation*}
\begin{proof}
To prove the result announced in the proposition, one should remark first that we can write the trace on the left-hand side of equation \eqref{eq: expectation of moments} 
 with the well known replica trick as:
\begin{equation*}
    \Tr(\rho_A^n)=\Tr\Big(\ketbra{\psi_G}{\psi_G}^{\otimes n}U_{\gamma_A}\otimes \id_B\Big),
\end{equation*}
The trace in the left-hand is on $\mathcal H_A$ that one rewrite as a full trace one $n$ copy of the full Hilbert space, bulk and boundary Hilbert space,  in the right-hand side of the equation above. Remark that we have used the notation $U_{\gamma_A}=\bigotimes_{(x,\cdot)\in A}U_{\gamma_{x}}$ the tensor product of unitary representation of the permutation $\gamma_{x}=(n\ldots 1)\in \mathcal S_n$ for each half edges $(x,\cdot)\in A$. 

By expanding and taking the average over random Gaussian states one obtains: 
\begin{align*}
    \E \Tr(\rho_A^n)&=\Tr\Big(\E \Big[\ketbra{\psi_G}{\psi_G}^{\otimes n}\Big]\,U_{\gamma_A}\otimes \id_B\Big)\\
    &=\Tr\Big(\bigotimes_{e\in E_b}\ketbra{\Omega_e}{\Omega_e}^{\otimes n}\E \Big[\bigotimes_{x\in V}\ketbra{g_x}{g_x}^{\otimes n}\Big]\,U_{\gamma_A}\otimes \id_B\Big)\\
    &=\Tr\Big(\bigotimes_{e\in E_b}\ketbra{\Omega_e}{\Omega_e}^{\otimes n}\bigotimes_{x\in V}\E \Big[\ketbra{g_x}{g_x}^{\otimes n}\Big]\,U_{\gamma_A}\Big), 
\end{align*}
where in the last equation above, we have used the shorthand notation $U_{\gamma_A}$ instead of $U_{\gamma_A}\otimes \id_B$.
We recall the following property of random Gaussian states see \cite{harrow2013church}:
\begin{equation*}
\forall x\in V,\quad\E\Big[\ketbra{g_x}{g_x}^{\otimes n}\Big]=\sum_{\{\alpha_{x}\}\in\mathcal{S}_n}U_{\alpha_{x}},
\end{equation*}
with $U_{\alpha_x}$ the unitary representation of $\alpha_x\in\mathcal S_n$. Each permutation $\alpha_x\in\mathcal{S}_n$ acts on each vertex Hilbert, hence implicitly on each edges associated to each vertex $x\in V$.  Therefore, the moments' formula becomes: 
\begin{align*}
    \E \Tr(\rho_A^n)&=\sum_{\{\alpha_x\}\in\mathcal{S}_n}\Tr\Big[\bigotimes_{e\in E_b}\ketbra{\Omega_e}{\Omega_e}^{\otimes n}\bigotimes_{x\in V}U_{\alpha_{x}}\,U_{\gamma_A}\Big]\\
    &=D^{-n|E_b|}\sum_{\{\alpha_x\}\in\mathcal S_n}\prod_{(x,\cdot)\in A}D^{\#(\gamma_x^{-1}\alpha_x)}\,\prod_{(x,\cdot)\in B}D^{\#(\id_x^{-1}\alpha_{x})}\prod_{(x,y)\in E_b}D^{\#(\alpha_x^{-1}\alpha_y)},
\end{align*}

where the formula above counts the number of loops obtained by contracting the maximally entangled states (edges) when one takes the trace. The factor of \(D^{-n|E_b|}\) appears due to the consequence of contracting the bulk edges, where each bulk edge contracted with itself, contributes a factor of \(D^{-1}\). By using the relation between the Cayley distance and the number of loops, we obtain the result in the statement of the proposition.  

\end{proof}

Graphically, one can understand the formula using Figure \ref{fig:wick} where we consider the case for \(n=3\). Upon utilizing the graphical integration technique for Wick integrals as presented in \cite{collins2011gaussianization,collins2016random}. We obtain loops and, consequently, Cayley distances of three kinds, (a) between $\id_x$ and elements directly connected to it, from the region $B$, (b) between $\gamma_x$ and elements directly connected to it, from the region \(A\) and (c) elements neither directly connected to $\id$ nor $\gamma$, from the bulk. Following this, we can rewrite the Hamiltonian in terms of Cayley distances as:

\begin{equation}\label{eq: Hamiltonian in G}
    H_G^{(n)}(\alpha):=\sum_{(x,\cdot) \in A}|\gamma^{-1}_x\alpha_x|+\sum_{(x,\cdot)\in B}|\id^{-1}_x \alpha_x|+\sum_{(x,y)\in E_b}|\alpha_x^{-1}\alpha_y|,
\end{equation}
 where $(x,\cdot)\in A$ represents half-edges in \(A\), $(x,\cdot)\in B$, represents half-edges in \(B\) and $(x,y)\in E_b$ represents edges in the bulk of the tensor network.

In the proposition above, we have addressed only the numerator term of the normalised quantum state $\tilde{\rho}_A$. However if one wants to compute the von Neumann and R\'enyi entropy (see equations \eqref{VN entropy} and \eqref{renyi entropy}), one should normalise the state and compute the moment. 

The following proposition gives the moment computation of the normalisation term in $\tilde{\rho}_A$. 

\begin{proposition}\label{prop: denominator contribution}
    For any $A\subseteq E_{\partial}$, we have 

\begin{equation}\label{eq: expectation of tensor moments}
    \E \left[\left(\Tr \rho_A\right)^n\right]=\sum_{\alpha=(\alpha_{x})\in\mathcal{S}_n^{|V|}} D^{{{n|E| - n|E_b|}}-h_G^{(n)}(\alpha)}, \;\forall\; n\in\mathbb N
\end{equation}
where the Hamiltonian $h_G^{(n)}(\alpha)$ is given by: 
\begin{equation}\label{eq: Hamiltonian with empty set}
h_G^{(n)}(\alpha):=\sum_{(x,\cdot) \in E_{\partial}}|\id_{x}^{-1}\alpha_x|+\sum_{(x,y)\in E_b}|\alpha_x^{-1}\alpha_y|.
 \end{equation}
\end{proposition}
\begin{proof}
The proof of this Proposition is a direct consequence of Proposition \ref{prop: average of moments} when one takes $A=\emptyset$, hence we obtain $h_G^{(n)}(\alpha)$ in the particular case when $A=\emptyset$ in $H_G^{(n)}(\alpha)$.
\end{proof}

\begin{figure}[htb] 
    \centering
    \includegraphics[width=0.7\textwidth]{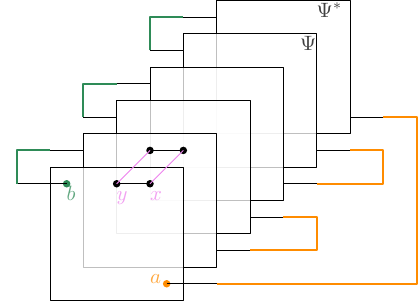}
    \caption{Graphical representation of the Wick theorem for the moment computation with $n=3$.}
    \label{fig:wick}
\end{figure}

\section{Asymptotic behaviour of moments}\label{sec: Optimisation with the MF}
This section will consist of describing the leading contributing terms as $D\to \infty$ of the moment by using the \emph{(maximal)-flow} approach. We will first introduce the (maximal)-flow approach wich will allows us to estimate the leading terms of the moments as $D\to\infty$ we refer to Proposition \ref{prop:min-H-max-flow} for more details. This result will allow us to deduce the convergence of the moment as $D\to \infty$ to moments of a graph dependent measure $\mu_{G_{A|B}}$ we refer to Theorem \ref{thm:limit-moments-general} for more details.

We recall first the obtained results from the previous section. In Proposition \ref{prop: average of moments} we have shown that the moments are given by:
\begin{equation*}
    \E \Tr (\rho_A^n)=\sum_{\alpha=(\alpha_x)\in\mathcal{S}^{|V|}_n}D^{n|E| - n|E_b| -H_G^{(n)}(\alpha)}
\end{equation*}
where the spin valued Hamiltonian in the permutation group $\mathcal S_n$ is given by:
\begin{equation*}
    H_G^{(n)}(\alpha):=\sum_{(x,\cdot) \in A}|\gamma_x^{-1}\alpha_x|+\sum_{(x,\cdot)\in B}|\id_x^{-1} \alpha_x|+\sum_{(x,y)\in E_b}|\alpha_x^{-1}\alpha_y|.
\end{equation*}
In particular, the contribution of the normalisation term in $\tilde \rho_A$ (see equation \eqref{eq: large bond dimension average entanglement entropie}) is the extended Hamiltonian $h_{G}^{(n)}(\alpha)$ as shown in Proposition \ref{prop: denominator contribution} when one takes $A=\emptyset$ in $H_G^{(n)}(\alpha)$. 
\bigskip

The main goal of this section, will consist on analysing the main contributed terms of the moment as $D\to \infty$. The leading terms will consist on solving the minimisation problem:
\begin{equation*}
    \min_{\alpha\in \mathcal{S}_n^{|V|}}H_G^{(n)}(\alpha).
\end{equation*}
Particularly as a consequence, we will minimize $h_G^{(n)}(\alpha)$ which will give us the leading contributed term as $D\to\infty$ of the normalisation term of $\tilde{\rho}_A$. The minimisation problem addressed above, will allow us to deduce the moment convergence as $D\to\infty$ to the moment of graph dependent measure $\mu_{G_{A|B}}$ in Theorem \ref{thm:limit-moments-general}. 

The minimisation problem above will be addressed with the \emph{(maximal)-flow} approach. This approach will consist first by constructing from the original graph $G$ a \emph{network} $G_{A|B}$. This network is constructed by adding first two extra vertices $\gamma$ and $\id$ to $G$ in such a way that all the half edges associated to $A$ are connected to the total cycle $\gamma$, and half edges in $B$ are connected to $\id$. The network $G_{A|B}$ has the same bulk structure of $G$, with the difference that all the vertices in $G_{A|B}$ are valued in the permutation group $\mathcal S_n$. 

The flow approach will consist on searching of different augmenting paths in the network $G_{A|B}$ that will start from $\id$ and ends to $\gamma$. This different paths will induce an \emph{order} structure in $G_{A|B}$. By taking off all the augmenting paths in $G_{A|B}$, we can find a lower bound of $H_{G_{A|B}}^{(n)}(\alpha)$ the extended Hamiltonian in the network $G_{A|B}$, we refer to Proposition \ref{prop: lower bound hamiltonian} for more details. Moreover, we will show that the minimum will be attained when the maximal flow starting from $\id$ to $\gamma$ is achieved, see Proposition \ref{prop:min-H-max-flow}. In particular we will show that the minimum of the extended Hamiltonian $h_{G_{A|B}}^{(n)}(\alpha)$ is zero, see Proposition \ref{prop: minimal of small hamiltonian} for more details.

Before we start with our flow approach, one should mention that the contributed terms of the moments at large dimension were analysed with the \emph{(minimal) cut} approach in \cite{cheng2022random}. The authors assumed the existence of two  \emph{disjoint minimal cut} in the graph separating the region of interest and the rest of the graph that will contribute in large bond dimension. With the maximal flow approach, that we will introduce, \emph{we do not assume any (minimal) cut assumption}. By identifying different augmenting paths achieving the maximal flow and uses the famous maximal-flow minimal-cut theorem (see e.g.~\cite[Theorem 8.6]{korte2011combinatorial}) one can deduce the different minimal cuts without any assumption.    

\begin{definition}\label{def: network tilde G}
    Let the network $G_{A|B}=(\tilde V,\tilde E)$ defined from the initial graph $G=(V,E)$ such that: 
    \begin{equation*}
    \tilde V:=V\sqcup \{\id,\gamma\}\quad\text{and}\quad\tilde E:= E_{ \tilde A}\sqcup E_b \sqcup E_{\tilde B}\end{equation*} 
   where the region $E_{\tilde A}$ and $E_{\tilde B}$ are defined as:
   \begin{align*}
       E_{ \tilde A}&:=\bigsqcup_{x\in V_A}(x,\gamma)\\
       E_{\tilde B}&:=\bigsqcup_{x\in V_B}(\id,x),
   \end{align*}
   where $V_A$ and $V_B$ denotes respectively all the vertices associated to the boundary region $A$ and $B$. Moreover the vertices are valued in the permutation group $\mathcal S_n$ where: $$\forall x\in \tilde V\to \alpha_x\in \mathcal S_n.$$
   \end{definition}
   Remark in the definition given above, the graph $G_{A|B}$ is constructed in such a way all the half edges $(x,\cdot)\in A$ are connected to $\gamma=(n\cdots 1)\in\mathcal S_n$ and the half edges $(x,\cdot)\in B$ are connected to $\id$. Note also that in $G_{A|B}$ there is no half edges, the bulk region in the network $G_{A|B}$ remains the same as the one in the graph $G$.
   
Let first consider the extended Hamiltonian $H_{G_{A|B}}^{(n)}(\alpha)$ of $H_{ G}^{(n)}(\alpha)$ in the network $G_{A|B}$ given by: 
\begin{equation}\label{eq: Hamiltonian in the graph G tilde}
    H_{G_{A|B}}^{(n)}(\alpha):=\sum_{x \in V_A}|\gamma^{-1}\alpha_x|+\sum_{x\in V_B}|\id^{-1} \alpha_x|+\sum_{(x,y)\in V_b}|\alpha_x^{-1}\alpha_y|,
\end{equation}

where each term in the new Hamiltonian is valued in the network $G_{A|B}$. Moreover, the sums in the above formula are over the vertices $V_A$, $V_B$ and $V_b$ are the vertices with the respective half edges in the region $A$, $B$ and $E_b$.

As was mentioned earlier, the flow approach will consist on analysing different paths that start from $\id$ and ends in $\gamma$. This will induce a natural orientation of the network $G_{A|B}$, more precisely a \emph{poset structure}. In the following, we will define the set of different paths in $G_{A|B}$ and the edges' disjoint paths.
\begin{definition}
 Let $\mathcal{P}(G_{A|B})$ be the set of all possible paths from the source to the sink in $G_{A|B}$, where the source and the sink in our case are the $\id$ and $\gamma$ respectively. Formally, the set of paths $\mathcal P(G_{A|B})$ is defined as:  
\begin{equation*}
    \mathcal{P}(G_{A|B}):=\{\pi_i:\quad \pi_i:\,\id\to \gamma\},
\end{equation*}
where $\{\pi_i\}_i$ are all the paths connecting the $\id$ to $\gamma$.
\end{definition}
\begin{definition}
Let $\tilde {\mathcal P}(G_{A|B})$ the set of all disjoint paths in $\mathcal P(G_{A|B})$,
\begin{equation*}
    \tilde {\mathcal P}(G_{A|B}):=\{\pi_i\in\mathcal{P}(G_{A|B}): \{\pi_i\}_i\, \text{ are edges disjoint}\,\}
\end{equation*}
\end{definition}
\begin{remark}
    It is clear from the definition that $\tilde{\mathcal{P}}(G_{A|B})\subseteq \mathcal{P}(G_{A|B})$.
\end{remark}
Searching for different paths that starts from the $\id$ and ends to $\gamma$ will induce an ordering, more precisely a \emph{poset structure} in the network $G_{A|B}$. First, we shall give in the following definition of a poset structure that will allow us later to use it in our maximal flow approach to minimize $H_{G_{A|B}}^{(n)}(\alpha)$. 
 
\begin{definition}
    The poset structure $\mathcal{P}_o(G_{A|B})$ is a \emph{homogeneous} relation denoted by $\leq$ satisfying the following conditions: 
    \begin{itemize}
        \item Reflexivity: $\alpha_x\leq \alpha_x$.
        \item  Antisymmetry: $\alpha_{x}\leq \alpha_{y}$ and $\alpha_{y}\leq \alpha_{x}$ implies $\alpha_{x}=\alpha_{y}$.
        \item Transitivity: $\alpha_{x}\leq\alpha_{y}$ and $\alpha_{y}\leq\alpha_{z}$ implies $\alpha_{x}\leq \alpha_{z}$.
    \end{itemize}
    for all $\alpha_x,\alpha_y,\alpha_z\in \tilde V$.
\end{definition}

\begin{definition}
Define the \emph{natural ordering} as:
\begin{equation*}
    \id\leq \alpha_1\leq\alpha_2\leq\cdots\leq\alpha_n\leq \gamma,
\end{equation*}
for a path $\pi_i\in\mathcal P(G_{A|B})$ given by 
\begin{equation*}
    \pi_i:\id\to \alpha_1\to\alpha_2\to\cdots\to\alpha_n\to\gamma.
\end{equation*}
\end{definition}
Another notion useful in our (maximal) flow analysis, is the \emph{permutation cluster}. We define a permutation cluster of a given permutation $\alpha_x$ as all the edge-connected permutations to $\alpha_x$.
\begin{definition}\label{def: Cluster}
   A \emph{permutation cluster} $[\alpha_x]$ is defined as all the edge-connected permutations to $\alpha_x \in\mathcal S_n$.
\end{definition}
\begin{remark}
With the poset structure in $G_{A|B}$, we have a naturally induced ordering in the cluster structures for each connected permutations to the permutation elemnents $\{\alpha_i\}_{i\in\{x,y,z\}}$ where all the properties of the above definition can be extended to the cluster $[\alpha_i]$ of a given permutation $\alpha_i$. More precisely the following holds: 
\begin{align*}
\alpha_x\leq \alpha_y&\implies [\alpha_x]\leq [\alpha_y].\\
\alpha_x=\alpha_y&\implies [\alpha_x]=[\alpha_y].&\\
\alpha_x\leq\alpha_y\leq\alpha_z&\implies [\alpha_x]\leq [\alpha_y]\leq [\alpha_z].
\end{align*}

\end{remark}

\begin{definition}
The maxflow in $G_{A|B}$ is the maximum of all the edges disjoint paths in $\tilde{\mathcal P}(G_{A|B})$:
\begin{equation*}
    \mathrm{maxflow}(G_{A|B}):=\max \left\{\left| \tilde{\mathcal{P}}(G_{A|B})\right| \, :\,\text{ s.t. the paths in } \tilde{\mathcal{P}}(G_{A|B}) \text{ are edge-disjoint}\right\}. 
\end{equation*}  
\end{definition}

In the following proposition, we will give a lower bound of the extended Hamiltonian $H_{G_{A|B}}(\alpha)$ which will be saturated when the maximal flow in $G_{A|B}$ is achieved as shown in Proposition \ref{prop:min-H-max-flow}. Hastings uses similar ideas in \cite[Lemma 4]{hastings2017asymptotics} to lower bound the moments of a random tensor network map.

\begin{proposition}\label{prop: lower bound hamiltonian}
Let $\alpha\in\mathcal{S}_n^{|V|}$ and $\tilde{\mathcal{P}}(G_{A|B})$ be an arbitrary set of edge-disjoint paths in $G_{A|B}$, and set $k:=|\tilde{\mathcal P}(G_{A|B})|$, the following inequalities holds:
\begin{equation}\label{eq: flow inequality}
    H_{G_{A|B}}^{(n)}(\alpha)\geq k(n-1)+H_{G_{A|B}\setminus \bigsqcup_{i\in[k]}\pi_i}^{(n)}(\alpha)\geq k(n-1),
\end{equation}
where $H_{G_{A|B}\setminus \bigsqcup_{i\in[k]}\pi_i}^{(n)}(\alpha)$ defined as: 
\begin{equation}\label{eq: Hamiltonian in  tilde G'}
    H_{G_{A|B}\setminus \bigsqcup_{i\in[k]}\pi_i}^{(n)}(\alpha):=\sum_{x\in V_A\setminus\bigsqcup_{i\in[k]}\pi_i}|\gamma^{-1}\alpha_x|+\sum_{x\in V_B\setminus\bigsqcup_{i\in[k]}\pi_i}|\id^{-1}\alpha_x|+\sum_{x\sim y\in V_b\setminus\bigsqcup_{i\in[k]}\pi_i}|\alpha_y^{-1}\alpha_x|.
\end{equation}

\end{proposition}
One should mention that in the proposition above the sums are over $\beta\setminus \bigsqcup_{i\in[k]}\pi_i$ for $\beta \in\{V_A,V_B,V_b\}$ which are the set of the different boundary and bulk regions when one takes off all the different edges and vertices that will contribute in different paths $\pi_i\in\tilde{\mathcal P}(G_{A|B})$ in $G_{A|B}$. 

\begin{proof}
    Let the set of edge disjoint paths $\{\pi_i\}_{i\in[k]}\in \tilde{\mathcal P}(G_{A|B})$. Fix a path $\pi_i$ for a given $i\in[k]$ where: 
    \begin{equation*}
        \pi_i:\id\to\alpha_{x_1}\to\alpha_{x_2}\to\cdots\to\alpha_{x_n}\to\gamma,    \end{equation*}
        is a path that starts from $\id$ and explores $\{x_i\}_{i\in[n]}$ vertices and ends in $\gamma$.
        By using equation $\eqref{eq: Hamiltonian in the graph G tilde}$, and using the path defined above one obtains: 
        \begin{equation*}
         H_{G_{A|B}}^{(n)}(\alpha)=|\alpha_{x_1}|+\sum_{i=1}^{n-1}|\alpha_{x_i}^{-1}\alpha_{x_{i+1}}|+|\alpha_{x_n}^{-1}\gamma|+H_{G_{A|B}\setminus \pi_i}^{(n)}(\alpha)\geq n-1+ H_{G_{A|B}\setminus \pi_i}^{(n)}(\alpha),
        \end{equation*}
        where we have used the triangle inequality of the Cayley distance and $|\gamma|=n-1$. The Hamiltonian $H_{G_{A|B}\setminus\pi_i}^{(n)}(\alpha)$ is the contribution when the  path $\pi_i$ from $G_{A|B}$ is used.

        By iteration over all the edges disjoint paths $\{\pi_i\}_{i\in[k]}\in \tilde{\mathcal P}(G_{A|B})$ one obtains the desired result. The second inequality is obtained by observing that $H_{G_{A|B}\setminus\pi_i}^{(n)}(\alpha)\geq 0$, ending the proof of the proposition.
\end{proof}

\begin{proposition}
    Given a graph $G$, there exist a tuple of permutations $\alpha$ such that $H_{G_{A|B}}^{(n)}(\alpha) = \operatorname{maxflow}(G_{A|B})$.
\end{proposition}

\begin{proof}
    By the celebrated max-flow min-cut theorem, the maximum flow in the network is equal to its minimal cut. Recall that a \emph{cut} of a network is a partition of its set of vertices into two subsets $S \ni s$ and $T \ni t$, and the size of the cut is the number of $S-T$ edges. In our setting, the max-flow min-cut theorem (see e.g.~\cite[Theorem 8.6]{korte2011combinatorial}) implies that there exists a partition of the vertex set $\tilde V$ of $G_{A|B}$ (see \cref{def: network tilde G} into two subsets, $\tilde V = S \sqcup T$, with $\id \in S$ and $\gamma \in T$, such that 
    $$\operatorname{maxflow}(G_{A|B}) = \Big | \{ (x,y) \in \tilde E \, : \, x \in S \text{ and } y \in T\} \Big |.$$

    Define, for $x \in V$, 
    $$\alpha_x = \begin{cases}
        \id &\qquad \text{ if } x \in S\\
        \gamma &\qquad \text{ if } x \in T.
    \end{cases}$$
    Since $\id \in S$ and $\gamma \in T$, we have:
    \begin{align*}
        H_{G_{A|B}}^{(n)}(\alpha) &= \sum_{x \in V_A}|\gamma^{-1}\alpha_x|+\sum_{x\in V_B}|\id^{-1} \alpha_x|+\sum_{(x,y)\in V_b}|\alpha_x^{-1}\alpha_y| \\
        &= \sum_{\substack{x\in V_A \\ x \in S}}|\gamma^{-1}\alpha_x|+\sum_{\substack{x\in V_B \\ x \in T}}|\id^{-1} \alpha_x|+\sum_{\substack{(x,y)\in V_b \\ x \in S \text{ and } y \in T}}|
        \alpha_x^{-1}\alpha_y| \\
        &= (n-1) \Big[ |\{ (x,\cdot) \in A \, : \, x \in S\}| + |\{ (x,\cdot) \in B \, : \, x \in T\}| + \\
        &\qquad\qquad\qquad |\{ (x,y) \in E_b \, : \, x \in S \text{ and } y \in T\}| \Big] \\
        &= (n-1)\operatorname{maxflow}(G_{A|B}),
    \end{align*}
    where we have used in the last claim the fact that there are no edges between $\id$ and $\gamma$ in $\tilde E$, see \cref{fig:cut-size}. 
    \begin{figure}
        \centering
        \includegraphics{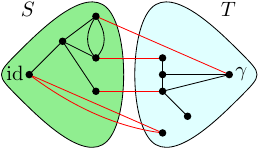}
        \caption{The different edges that correspond to a cut $S-T$ in the network $G_{A|B}$.}
        \label{fig:cut-size}
    \end{figure}
\end{proof}

\begin{proposition}\label{prop:min-H-max-flow}
For all $n \geq 1$, we have

\begin{equation*}
    \min_{\alpha \in \mathcal S_n^{|V|}} H_{G_{A|B}}^{(n)}(\alpha)=(n-1)\operatorname{maxflow}(G_{A|B}).
\end{equation*}
\end{proposition}
\begin{proof}
    This follows from the two previous propositions. 
\end{proof}

Once we identify and remove all the augmenting paths in the network $G_{A|B}$ achieving the maximal flow, we obtain a \emph{clustered graph} $G_{A|B}^c$ by identifying different remaining connected permutations.

The following example gives an illustration of the different steps described above to analyse the maxflow problem in the case of the tensor network represented in Figure \ref{fig:tensor-network-with-marginal}. 
\begin{example}
Figure \ref{fig:network} represents the network $G_{A|B}$ associated with the random tensor network from Figure \ref{fig:tensor-network-with-marginal}. The vertices in the network are valued in the permutation group $\mathcal S_n$. The network is constructed by adding two extra vertices $\gamma$ and $\id$ by connecting all the half edges in $A$ to $\gamma$ and the half edges in $B$ to $\id$. The flow approach induces a flow from $\id$ to $\gamma$ where the maximum flow in Figure \ref{fig:network} is $4$ where the augmenting paths achieving it are colored. By removing the four edge-disjoint augmenting paths we obtain the clustered graph $G_{A|B}^c$ in Figure \ref{fig:residual-network} by identifying the remaining connected edges as a single permutation cluster, i.e $\id$ with $\alpha_{15}$ to form the cluster $[\id,15]$.  
\end{example}

\begin{theorem}\label{thm:limit-moments-general}
    In the limit $D \to \infty$, we have, for all $n \geq 1$, 
    \begin{equation*}
        \lim_{D\to \infty}\E\frac{1}{D^{F(G_{A|B})}}\Big[\Tr\left(\left(D^{F(G_{A|B})-|E_{\partial}|}\,\rho_A\right)^n\right)\Big]=m_{n},
    \end{equation*}
    where $m_n$ is the number of permutations achieving the minimum of the network Hamiltonian $G_{A|B}^{(n)}$. These numbers are the moments of a probability measure $\mu_{G_{A|B}}$ which is compactly supported on $[0, +\infty)$.
    
\end{theorem}
\begin{proof}
   For fixed $n$, the convergence to $m_n$, the number of minimizers of the Hamiltonian $G_{A|B}^{(n)}$, follows from \cref{prop: average of moments} and \cref{prop:min-H-max-flow}. The claim that the numbers $(m_n)_n$ are the moments of a compactly supported probability measure follows basically from Prokhorov's theorem \cite[Section 5]{billingsley2013convergence} (see also \cite[Footnote 2]{hastings2017asymptotics}). Indeed, note that, at fixed $D$, the quantity 
    $$\E\frac{1}{D^{F(G_{A|B})}}\Big[\Tr\left(\left(D^{F(G_{A|B})-|E_{\partial}|}\,\rho_A\right)^n\right)\Big]$$
    is the $n$-th moment of the empirical eigenvalue distribution of the random matrix $D^{F(G_{A|B})-|E_{\partial}|}\,\rho_A$, restricted to a subspace of dimension $D^{F(G_{A|B})}$ containing its support (this follows from the fact that $D^{F(G_{A|B})}$ is an upper bound on the rank of $\rho_A$, see \cref{eq:upper-bound-rank}). These measures have finite second moment, so the sequence (index by $D$) is tight. The limiting moments satisfy Carleman's condition since $m_n \leq \mathrm{Cat}_{n}^{|V|}$, proving that the limit measure $\mu_{G_{A|B}}$ has compact support; recall that $\mathrm{Cat}_n \leq 4^n$ is the $n$-th \emph{Catalan number}, see \cref{sec: appendix}. Since the matrix $\rho_A$ is positive semidefinite, $\mu_{G_{A|B}}$ must be supported on $[0,+\infty)$. 
\end{proof}
\begin{remark}
The obtained moments are given by a graph dependent measure. We will show in the following sections that such measures can be explicitly constructed if the partial order $G_{A|B}^o$ is series-parallel (see Section \ref{sec: entanglement and free probability theory} and Theorem \ref{Th: moments and free pro} for more details).
\end{remark}

In all that we have described above, the contribution terms at large bond dimension $D\to \infty$ of $\E\Tr(\rho_A^n)$ are the ones that minimise $H_{G_{A|B}}^{(n)}(\alpha)$. As we have shown in Proposition \ref{prop:min-H-max-flow} $H_{G_{A|B}}^{(n)}(\alpha)$ is minimized when the maximal flow is attained in $G_{A|B}$.

For later purposes, if one wants to analyse the moment of $\tilde{\rho}_A$, one should also consider the contribution of the normalisation term of $\tilde{\rho}_A$ at large bond dimension. We recall from Proposition \ref{prop: denominator contribution} the contribution of the normalisation term is given by: 
\begin{equation*}
    \E \left[(\Tr \rho_A)^n\right]=\sum_{\alpha=(\alpha_{x})\in\mathcal{S}_n^{|V|}} D^{{{n|E| - n|E_b|}}-h_{G}^{(n)}(\alpha)}, \;\forall\; n\in\mathbb N
\end{equation*}
where 
\begin{equation*}
h_{G}^{(n)}(\alpha):=\sum_{(x,\cdot) \in E_{\partial}}|\id^{-1}_x\alpha_x|+\sum_{(x,y)\in E_b}|\alpha_x^{-1}\alpha_y|.
 \end{equation*}
At large dimension $D\to \infty$, the contributed terms are given by the one that will minimize the extended Hamiltonian $h_{G_{A|B}}^{(n)}(\alpha)$ in $G_{A|B}$: 
\begin{equation*}
    h_{G_{A|B}}^{(n)}(\alpha):=\sum_{x \in V_{\partial}}|\id^{-1}\alpha_x|+\sum_{(x,y)\in V_b}|\alpha_x^{-1}\alpha_y|,
\end{equation*}
where the first some is over all the vertices $V_{\partial}$ with boundary edges, and $V_b$ are the bulk vertices. 
\begin{proposition}\label{prop: minimal of small hamiltonian}
Let $h_{G_{A|B}}^{(n)}(\alpha)$ the extended Hamiltonian in $G_{A|B}$. For all $n\geq 1$, we have: 
\begin{equation*}
\min_{\alpha\in \mathcal S_n^{|\tilde V|}}h_{G_{A|B}}^{(n)}(\alpha)=0,
\end{equation*}
achieved by identifying all the permutations with $\id$.
\end{proposition}
\begin{proof}
    To minimize the Hamiltonian $h_{G_{A|B}}^{(n)}(\alpha)$ we shall follow the same recipe where we connect all half edges $A$ to $\gamma$ and half edges $B$ to $\id$. However in $h_{G_{A|B}}^{(n)}(\alpha)$, all the boundary terms will be connected to $\id$, hence no path starts from $\id$ that ends in $\gamma$. By the bulk connectivity of $G$, the minimum is achieved by identifying all the permutations to $\id$, therefore by Proposition \ref{prop:min-H-max-flow} we obtain the desired result.  
\end{proof}
\begin{remark}
    In Proposition \ref{prop: minimal of small hamiltonian}, the Hamiltonian $h_{G_{A|B}}^{(n)}(\alpha)$ is obtained by tacking $A=\emptyset$ in $H_{G}^{(n)}(\alpha)$ (see equation \eqref{eq: Hamiltonian}). One should mention if $B=E_{\partial}\setminus A=\emptyset$ we will have the same form of the Hamiltonian $h_{G_{A|B}}^{(n)}(\alpha)$ where instead of all the half edges connected to $\id$ they will be all connected to $\gamma$. Therefore one deduce that there is no paths that starts from $\id$ and ends to $\gamma$, hence the minimum is $0$ achieved by identifying all the permutations with $\gamma$.  
\end{remark}
\begin{corollary}\label{corr: average of normalisation}
    For any $A\subseteq E_{\partial}$ moments of the normalisation term converges to $1$, more precisely: 
    \begin{equation*}
\lim_{D\to\infty}\E\left(\Tr\left(D^{-|E_\partial|}\rho_A\right)\right)^n=1.
    \end{equation*}
\end{corollary}
\begin{proof}
    By tacking the average as was shown in Proposition \ref{prop: denominator contribution} one obtains the Hamiltonian $h_{G}^{(n)}(\alpha)$. By the maximal flow the Hamiltonian is minimised by identifying all the permutations to the $\id$, therefore $F(G_{A|B})=0$ as was shown in Proposition \ref{prop: minimal of small hamiltonian}. Therefore by removing all the augmenting paths achieving the maximal flow the obtained residual graph is trivial with only two disjoint vertices $\gamma$ and the identity cluster $[\id]$. Hence by Theorem \ref{thm:limit-moments-general} one obtains the desired result.
\end{proof}

\section{Moment for ordered series-parallel network}\label{sec: entanglement and free probability theory}
In this section, we will introduce the notion of a series-parallel graph. This notion will allow us to compute the moment as an explicit graph-dependent measure explicitly. More precisely we will show with the help of free probability in the case of the obtained partial order $G_{A|B}^o$ is series-parallel the obtained graph-dependent measure is explicitly constructed. 

In Subsection \ref{subsec: series parallel graph} we will introduce the notion of the series-parallel graph and the associated measures. In Subsection \ref{subs: moment as graph} we will show in the case of a series-parallel partial order $G_{A|B}^o$ the moments converge to moments of a graph-dependent measure.

\subsection{Series-Parallel graph}\label{subsec: series parallel graph}
In this subsection, we introduce the notion of the series-parallel partial orders which will allow us in the following subsection to explicitly compute the moments as graph dependent measures. 

We shall start first by recalling first the notion of series-parallel partial order \cite{bechet1997complete} and giving some crucial definitions that will play an important role in all the rest of this section.

Given two partial orders $(P_i, \leq_i)$, $i=1,2$, one defines their \emph{series}, resp.~\emph{parallel}, composition as follows. The base set is $P:= P_1 \sqcup P_2$ and the order relation is: $x \leq y$ if 
\begin{itemize}
    \item $x,y \in P_i$ and $x \leq_i y$ \emph{or} $x \in P_1$ and $y \in P_2$ in the series case;
    \item $x,y \in P_i$ and $x \leq_i y$ in the parallel case.
\end{itemize}

It is more convenient for us to represent partial orders by their \emph{covering graphs}, where to a partial order $(P, \leq)$ we associate an oriented graph $G(V,E)$, with $V=P$ and $x \to y \in E$ iff $x <y$ and $\nexists z$ s.t.~$x <z<y$. We recall that we write $x < y$ to denote $x \leq y$ and $x \neq y$. The series and parallel composition for partial orders have an elegant interpretation in terms of directed graphs (or networks in this case). In what follows, we shall interchangeably use the terms partial order or partial order graph. 
\begin{definition}\cite{bechet1997complete}\label{def: series parallel network}
Let $H_1$ and $H_2$ two directed graph with there respective source $s_i$ and sink $t_i$ for $i\in\{1,2\}$. A \emph{series-parallel network} is a directed graph $G=(V,E)$ containing two distinct vertices $s \neq t \in V$, called the \emph{source} and the \emph{sink} that can be obtained recursively from the trivial network $G_{\text{triv}} = (\{s,t\}, \{\{s,t\}\})$ using the following two operations:
\begin{itemize}
        \item Series concatenation: $G=H_1 \bigSsqcup H_2$ is obtained by identifying the sink of $H_1$ with the source of $H_2$, i.e $t_1=s_2$.
        \item Parallel concatenation: $G=H_1\bigPsqcup H_2$ obtained by identifying the source and the sink of $H_1$ and $H_2$, i.e.~$s_1=s_2$ and $t_1=t_2$.
    \end{itemize}
\end{definition}

\begin{figure}[htb]
    \centering
    \includegraphics[scale=0.75]{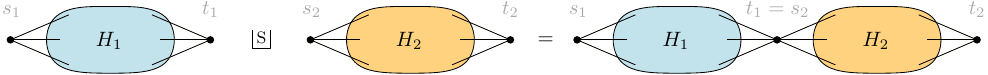}\\\bigskip
    \includegraphics[scale=0.75]{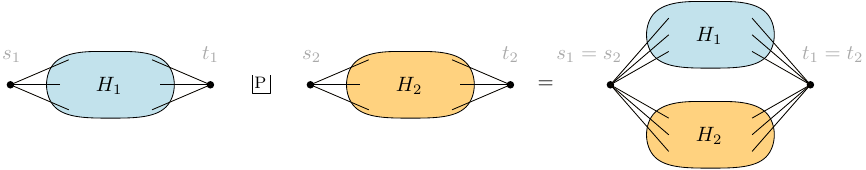}
    \caption{Series (top) and parallel (bottom) composition of two networks (graphs) $H_{1,2}$.}
    \label{fig:series-parallel-composition}
\end{figure}

\begin{remark}
Note that the parallel concatenation is a commutative operation, while the series concatenation is not, in general, commutative:
\begin{align*}
    \forall G,\,H \qquad G \bigPsqcup H &= H \bigPsqcup G \\
    \text{in general} \qquad G \bigSsqcup H &\neq H \bigSsqcup G.
\end{align*}
\end{remark}

We shall associate from a given series-parallel network different probability distributions constructed from the paralllel and the series concatenation introduced in Definition \ref{def: series parallel network}. 

\begin{definition}\label{def:series-parallel-probability-measure}
 To a series-parallel network $G$ we associate a probability measure $\mu_G$, defined recursively as follows: 
    \begin{itemize}
        \item To the trivial network $G_{\text{triv}} = (\{s,t\}, \{\{s,t\}\})$, we associate the Dirac mass at $1$:
        $$\mu_{G_{\text{triv}}} := \delta_1$$
        \item Series concatenation corresponds to the free multiplicative convolution of the parts, along with the measure $\MP$:
        $$\mu_{G \bigSsqcup H} := \mu_G \boxtimes \MP\boxtimes \mu_H$$
        \item Parallel concatenation corresponds to the classical multiplicative convolution of the parts:
        $$\mu_{G \bigPsqcup H} := \mu_G \times \mu_H.$$
    \end{itemize}
\end{definition}
In the definition above we have used the free product convolution $\boxtimes$ and the Mar\u{c}henko-Pastur distribution $\MP$. We refer to the Appendix \ref{sec: appendix} for a self-contained introduction to free probability theory.  
\subsection{Moment as graph dependent measure}\label{subs: moment as graph}
In this subsection, with the help of the series-parallel notion introduced in the previous subsection, we will show the moments $m_n$ in Theorem \ref{thm:limit-moments-general} are explicitly constructed from a graph-dependent measure in the case of the obtained partial order $G_{A|B}^o$ is series-parallel. 

Before we give the results of this subsection we recall first the different results obtained from the previous sections. From a given random tensor network as represented for an example in Figure \ref{fig:tensor-network-with-marginal}, we have computed in Section \ref{sec: moment computation} the moment for a normalised quantum state to a given subregion $A\subseteq E_{\partial}$ of the graph (see Propositions \ref{prop: average of moments} and \ref{prop: denominator contribution}). We approached the evaluation of the moment as $D\to\infty$ by the maximal flow approach as analysed in Section \ref{sec: Optimisation with the MF}. We have constructed from the graph $G$  the network $G_{A|B}$ by connecting each of the regions $A$ and $B$ respectively to $\gamma$ and $\id$. The flow consists of analysing the different paths starting from $\id$ and ending in $\gamma$. By taking off all the different augmenting paths achieving the maximal flow a clustered graph $G_{A|B}^c$ remains by identifying different edge-connected permutations. As represented in Figure \ref{fig:residual-network} for the clustered graph associated with the network $G_{A|B}$ in Figure \ref{fig:tensor-network-with-marginal}. With the maximal flow, we were able in Proposition \ref{prop:min-H-max-flow} which allows us to show the convergence of moments given by a graph dependent measure $\mu_{G_{A|B}}$ as shown in Theorem \ref{thm:limit-moments-general}. Moreover from Proposition \ref{prop: minimal of small hamiltonian} one deduce in Corollary \ref{corr: average of normalisation} that the normalisation terms converge to $1$. 

From the clustered graph $G_{A|B}^c$, we will construct an \emph{partial order} $G_{A|B}^o$ where the vertices in $G_{A|B}^o$ are the different permutation clusters. See Figure \ref{fig:partial-order-example} for the obtained partial order $G_{A|B}^o$ to the network $G_{A|B}$ in Figure \ref{fig:tensor-network-with-marginal}. If the partial order $G_{A|B}^o$ is series-parallel (see Definition \ref{def: series parallel network}), then we will explicitly show, in the following subsections, that we have a convergence in moments of $\tilde{\rho}_A$ to an explicit partial order measure $\mu_{G_{A|B}^o}$.

The following theorem shows the convergence to a moment-dependent measure $\mu_{G_{A|B}^o}$ in case of obtained partial order $G_{A|B}^o$ is series-parallel.

\begin{theorem}\label{Th: moments and free pro}
    For any $A\subseteq E_{\partial}$, and assuming the partial order $G_{A|B}^o$ is series-parallel, then the limit measure from \cref{thm:limit-moments-general} can be explicitly constructed from the partial order:
    $$\mu_{G_{A|B}} = \mu_{G^o_{A|B}}.$$
    In particular, the moments of the reduced tensor network matrix are given by:  
    \begin{equation}\label{eq:limit-moments-rho-A}
        \lim_{D\to\infty}\frac{1}{D^{F(G_{A|B})}}\E\Tr\left(\left(D^{F(G_{A|B})-|E_\partial|}\rho_A\right)^n\right) =\int t^n\,\mathrm{d}\mu_{G_{A|B}^o}.
    \end{equation}
\end{theorem}

\begin{proof}
    All we need to show is that the numbers $m_{n,G_{A|B}^o}$ are the moments of the probability measure $\mu_{G_{A|B}^o}$. We shall prove this using the recursive structure of the series-parallel networks (see Definition \ref{def: series parallel network}) and that of the probability measure $\mu_{G_{A|B}^o}$ (see Definition \ref{def:series-parallel-probability-measure}). 

    If the partial order $G_{A|B}^o$ is trivial, it consists only of two connected components, that of the identity (the source) $[\id]$ and that of the sink, $[\gamma]$. Hence, all the permutations associated to the connected components are fixed to be either $\id$ or $\gamma$. We have thus $m_{n, G_{A|B}^o} = 1$ for all $n \geq 1$, which are the moments of the measure $\mu_{G_{A|B}^o} = \delta_1$. This shows that the claim holds for the initial case of a trivial network. 

    If the partial order $G_{A|B}^o$ is the \emph{parallel} concatenation of two networks $G_{A|B}^o = H_1 \bigPsqcup H_2$ having the same source and sink as $G_{A|B}^o$, the geodesic equalities for $G_{A|B}^o$ are the disjoint union of the geodesic equalities for the vertices in $H_1$ and those for the vertices of $H_2$. This implies in turn that, for all $n \geq 1$,
    $$m_{n, G_{A|B}^o} = m_{n,H_1} \cdot m_{n,H_2},$$
    since there is no geodesic inequality mixing vertices from $H_1$ with vertices in $H_2$. Hence, by the induction hypothesis, we have
    $$m_{n, G_{A|B}^o} = \int t^n\, \mathrm{d}\mu_{H_1} \cdot \int t^n\, \mathrm{d}\mu_{H_2} = \int t^n\, \mathrm{d}\big[\mu_{H_1} \bigPsqcup \mu_{H_2}\big] = \int t^n\, \mathrm{d}\mu_{G_{A|B}^o},$$
    proving the claim for the parallel concatenation of networks. 

    Finally, let us consider the case where the network is the \emph{series} concatenation of two networks $G_{A|B}^o = H_1 \bigSsqcup H_2$. This means that there is a connected component, call it $[\beta]$ which is common of the two networks, being the sink of $H_1$ and the source of $H_2$. All geodesic equality conditions for the $H_1$ are of the form
    $$\id \to \alpha^{(1)}_1 \to \cdots \to \alpha^{(1)}_{k_1} \to \beta,$$
    while those of $H_2$ are of the form
    $$\beta \to \alpha^{(2)}_1 \to \cdots \to \alpha^{(2)}_{k_2} \to \gamma.$$
    In particular, the geodesic equality conditions for $G_{A|B}^o = H_1 \bigSsqcup H_2$ are of the form 
    $$\id \to \alpha^{(1)}_1 \to \cdots \to \alpha^{(1)}_{k_1} \to \beta \to \alpha^{(2)}_1 \to \cdots \to \alpha^{(2)}_{k_2} \to \gamma.$$
    The variable $\beta$ is a non-constrained non-crossing partition of $[n]$, and summing over it corresponds to taking the free multiplicative convolution with respect to the Mar\u{c}henko-Pastur distribution:
    $$m_{n,G_{A|B}^o} = \sum_{\substack{\beta \in \NC(n)\\ \alpha^{(1)}_i \leq \beta \leq \alpha^{(2)}_j}} 1 = \int t^n\, \mathrm{d} \mu_{H_1} \boxtimes \MP \boxtimes \mu_{H_2} = \int t^n\, \mathrm{d} \mu_{H_1 \bigSsqcup H_2} = \int t^n\, \mathrm{d}\mu_{G_{A|B}^o},$$
    proving the final claim and concluding the proof. 
\end{proof}

\begin{example}
    As an example, let us consider the graph illustrated in \cref{fig:tensor-network-with-marginal}. As we have described in the previous sections, the dominant terms of moments in Proposition \ref{prop: average of moments} are obtained by analyzing the maximum flow in $G_{A|B}$, given in \cref{fig:network} where $\operatorname{maxflow}(G_{A|B})=4$. The partial order $G_{A|B}^o$, obtained by removing from $G_{A|B}$ the edges that participate in the maximum flow is depicted in \cref{fig:residual-network}. Using the 4 augmenting paths (displayed in colors in \cref{fig:network}), we construct the \emph{partial order} on the connected components of the partial order, that we depict in \cref{fig:partial-order-example}.
    
    \begin{figure}[htpb]
        \centering
        \includegraphics{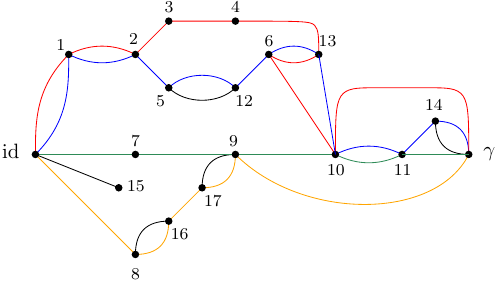}
        \caption{The network associated to the random tensor network marginal from \cref{fig:tensor-network-with-marginal}. The maximum flow of this network (with source $\id$ and sink $\gamma$) is 4, the four augmenting paths achieving this value are colored.}
        \label{fig:network}
    \end{figure} 
    
    \begin{figure}[htpb]
        \centering
        \includegraphics{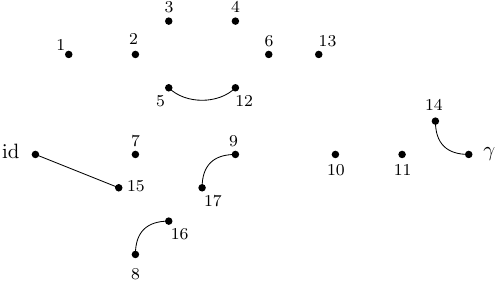}
        \caption{The clustered network corresponding to \cref{fig:network}, obtained by removing the four edge-disjoint augmenting paths.}
        \label{fig:residual-network}
    \end{figure}    
    
    \begin{figure}[htpb]
        \centering
        \includegraphics{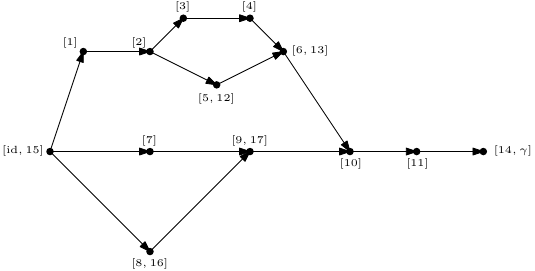} 
        \caption{The order graph corresponding to the networks from \label{fig: order graph}. The partial order relations are to be read from left to right. The elements of this partial order relation are the connected components of the clustered network from \cref{fig:residual-network}, eventually identified after taking into account the inequalities from the augmenting paths from \cref{fig:network}.}
        \label{fig:partial-order-example}
    \end{figure}    
    
    This process is fundamental in our approach, we give the details for one of these geodesics next. For example, consider the augmenting path 
    $$ \id \to 1 \to 2 \to 3 \to 4 \to 13 \to 6 \to 10 \to \gamma$$
    depicted in red in \cref{fig:network}. Since in the clustered graph from \cref{fig:residual-network} the respective pair of points $(\id, 15)$, and $14, \gamma$ are in the same connected components (clusters), this augmenting path gives rise to the following list of partial order relations: 
    $$[\id, 15] \preceq [1] \preceq [2] \preceq [3] \preceq [4] \preceq [13] \preceq [6] \preceq [10] \preceq [14,\gamma].$$
    The other three augmenting paths, depicted respectively in blue, green, and orange in \cref{fig:network}, give rise to the following list of inequalities: 
    \begin{align*}
        [\id, 15] &\preceq [1] \preceq [2] \preceq [5,12] \preceq [6] \preceq [13] \preceq [10] \preceq [11] \preceq [14,\gamma]\\
        [\id, 15] &\preceq [7] \preceq [9,17] \preceq [10] \preceq [11] \preceq [14,\gamma]\\
        [\id, 15] &\preceq [8,16] \preceq [9,17] \preceq [14,\gamma].
    \end{align*}
    The partial order depicted in \cref{fig:partial-order-example} is compiled from the set of inequalities coming from the (fixed) list of augmenting paths yielding the maximum flow (here 4). Note that, importantly, some connected components (clusters) can be identified in this partial order, due to the \emph{anti-symmetry} property $x \preceq y \text{ and } y \preceq x \implies x=y$; this happened in this example for the clusters $[6]$ and $[13]$. 
    
    As an application of Theorem \ref{Th: moments and free pro}, one can give explicit moments of the measure $\mu_{G_{A|B}^o,A}$: 
    \begin{equation}
\lim_{D\to\infty}\E D^{-4}\Tr\Big[(D^{-6}\rho_A)^n\Big]=m_{n,G_{A|B}^o}=\int x^n\,\mathrm{d}\mu_{G_{A|B}^o}.        
    \end{equation}
The powers of $D$ in the normalization follow from $|E_\partial| = 10$ (see the boundary edges in \cref{fig:tensor-network-with-marginal}) and from $\operatorname{maxflow}(G_{A|B}) = 4$. The resulting probability measure $\mu_{G_{A|B}^o}$ associated to the partial order from \cref{fig:partial-order-example} is given by:
\begin{equation}
    \mu_{G_{A|B}^o} = \left\{\left[ \MP^{\boxtimes 3} \boxtimes (\MP^{\boxtimes 2} \times \MP) \right] \times \left[ (\MP \times \MP) \boxtimes \MP \right] \right\} \boxtimes \MP^{\boxtimes 2}.
\end{equation}
The measure given above is obtained by the iterative procedure from  \cref{def:series-parallel-probability-measure} as follows. First, observe that the graph in \cref{fig:partial-order-example} can be decomposed as a \emph{series} composition of three graphs $G_1 \bigSsqcup G_2 \bigSsqcup G3$:
\begin{figure}
\centering
\includegraphics{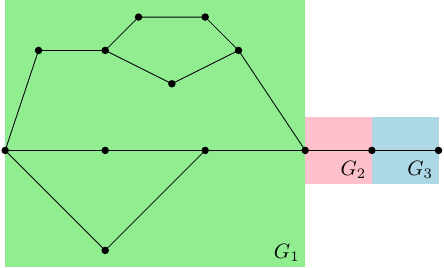}
\caption{Representation of the partial order $G_{A|B}^o$ factorises to a series combination of graphs $G_1$, $G_2$ and $G_3$.}\label{fig: G1G2G3}
\end{figure}
hence, using $\mu_{G_2} = \mu_{G_3} = \delta_1$, we have
$$ \mu_{G_{A|B}^o} = \mu_{G_1} \boxtimes \MP \boxtimes \mu_{G_2} \boxtimes \MP \boxtimes \mu_{G_3} =  \mu_{G_1} \boxtimes \MP^{\boxtimes 2}.$$
Observe now that $G_1$ is a \emph{parallel} composition of two other graphs
\begin{figure}[htb]
    \centering
    \includegraphics{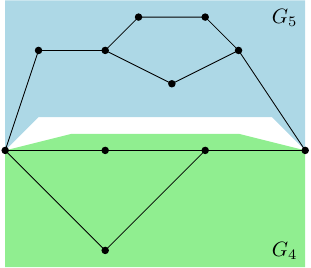}
    \caption{The graph $G_1$ is factorized to the parallel composition of $G_4$ and $G_5$.}\label{fig: G4G5}
\end{figure}
hence
$$\mu_{G_1} = \mu_{G_4} \times \mu_{G_5}.$$
Let us now analyze separately $G_4$ and $G_5$. Firstly, $G_4$ can be decomposed as a series composition between the parallel composition of $G_6$ and $G_7$, and $G_8$:
\begin{figure}    \includegraphics{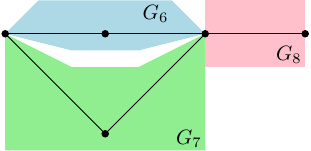}
    \caption{The graph $G_4$ is factorised to the series composition of the graph $G_8$ with the graph $G_6$ parallel to $G_7$.}
    \label{fig: G6G7G8}
\end{figure}
that is 
$$G_4 = \left( G_6 \bigPsqcup G_7 \right) \bigSsqcup G_8 \implies \mu_{G_4} = \left( \mu_{G_6} \times \mu_{G_7} \right) \boxtimes \MP \boxtimes \mu_{G_8}.$$
Now, $G_6$ and $G_7$ are series compositions of two trivial graphs, so $\mu_{G_6} = \mu_{G_7} = \MP$, while $\mu_{G_8} = \delta_1$. We have thus
$$\mu_{G_4} = \left( \MP \times \MP \right) \boxtimes \MP.$$
Let us now turn to $G_5$, which can be decomposed as follows: 
\begin{figure}
\centering
    \includegraphics{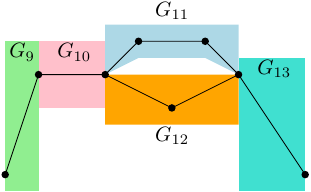}
    \caption{The graph $G_5$ factorises as a serie composition of $G_9$, $G_{10}$ with $G_{11}$ composed in parallel with $G_{12}$ and series with $G_{13}$.}\label{fig G5}
\end{figure}
that is 
$$G_5 = G_9 \bigSsqcup G_{10} \bigSsqcup \left( G_{11} \bigPsqcup G_{12} \right) \bigSsqcup G_{13}.$$
In terms of the associated probability measures, we have
$$\mu_{G_5} = \mu_{G_9} \boxtimes \MP \boxtimes \mu_{G_{10}} \boxtimes \MP \boxtimes \left( \mu_{G_{11}} \times \mu_{G_{12}} \right) \boxtimes \MP \boxtimes \mu_{G_{13}}.$$
Using iteratively series compositions, we have 
$$\mu_{G_{11}} = \MP^{\boxtimes 2} \qquad \text{and} \qquad \mu_{G_{12}} = \MP.$$
We obtain
$$\mu_{G_5} = \MP^{\boxtimes 3} \boxtimes \left( \MP^{\boxtimes 2} \times \MP\right).$$
Putting all these ingredients together, we obtained the announced formula for $\mu_{G_{A|B}^o}$.
\end{example}
\begin{remark}
In the example of the tensor network represented in Figure \ref{fig:tensor-network-with-marginal} we were able to compute the moments from the factorised series-parallel thought the flow approach. One should mention if one take the minimal cut approach to the problem, there exist minimal cuts in the network represented in Figure \ref{fig:network} do intersect, see \cref{fig:network-cuts}. Therefore we can compute the correction terms of the entropy as the moment of a given measure without any minimal cut assumption considered in previous work.

\begin{figure}[htb]
    \centering
    \includegraphics{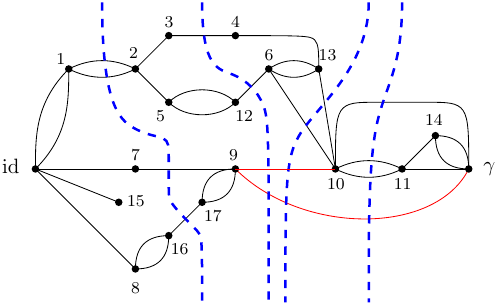}
    \caption{Different cuts of size 4 in the network of \cref{fig:network}, represented with blue, dashed lines. Notice that some of these minimal cuts (the maximum flow in the network is 4) share some edges, represented in red.}
    \label{fig:network-cuts}
\end{figure}

\end{remark}

\begin{remark}\label{rmq: measure mu ab compact support}
    The obtained measure $\mu_{G_{A|B}^o}$ for a given ordered series-parallel graph $G_{A|B}^o$ has a compact support where it combines the Mar\u{c}henko-Pastur distribution with classical product measure and free product convolution constructed from the structure of $G_{A|B}^o$.   
\end{remark}

\section{Examples of series-parallel networks}\label{sec:examples}

In this section we apply the results obtained previously for various random tensor networks having an induced series-parallel order. We start from simple cases and work our way towards more physically relevant cases. 

\subsection{Single vertex network}

We start with the simplest possible case: a tensor network having only one vertex, no bulk edges, and two boundary half-edges, see \cref{fig:single-vertex}. For this network, the associated random tensor 
$$\Psi_G \in \mathbb C^D \otimes \mathbb C^D$$
has i.i.d.~standard complex Gaussian entries. The two boundary half-edges are partitioned into two one-element sets $A$ and $B = \bar A$. From this tensor, we construct the reduced matrix 
$$\rho_A := \Tr_B \ketbra{\Psi_G}{\Psi_G}$$
obtained by partial tracing the half-edge $B$. Note that in this very simple case, the matrix $\rho_A$ can also be seen as a product of the matricization of the tensor $\Psi_G$ with its hermitian adjoint, hence $\rho_A$ is a Wishart random matrix (see \cref{sec: appendix} for the definition and basic properties of Wishart matrices). 

In order to analyze the large $D$ spectral properties of $\rho_A$, we first construct the network $G_{A|B}$, obtained by connecting all the half-edges in $G$ that belong to $A$ to a new vertex $\gamma$ and those in $B$ to a new vertex $\id$. The flow analysis of this network is trivial: there is a unique path from $\id$ to $\gamma$, hence the maximum flow is $1$ and the residual network is empty (both edges in the network have been used for the construction of the unique maximum flow). 

\begin{figure}[htb]
    \centering
    \includegraphics{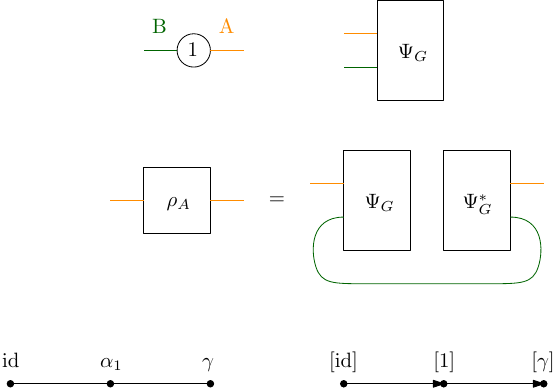}\bigskip
    \caption{A single vertex network. On the top row, we represent the network and the associated random tensor $\Psi_G$. In the middle row we represent the reduced matrix $\rho_A$, obtained by partial tracing the edge $B$ between $\Psi_G$ and $\Psi_G^*$. In the bottom row, we represent the network $G_{A|B}$ and the partial order graph $G_{A|B}^o$.}
    \label{fig:single-vertex}
\end{figure}

Since there is a unique path achieving maximum flow and a single vertex in the network, the partial order induced by the path is very simple: $\id - \alpha_1 - \gamma$. Hence, the only condition on the permutation $\alpha_1 \in \mathcal S_n$ is that it should lie on the geodesic between the identity permutation $\id$ and the full cycle permutation $\gamma \in \mathcal S_n$. We have thus a series network, see \cref{fig:single-vertex} bottom right diagram. The limit moment distribution is $\MP$, the Mar\u{c}henko-Pastur distribution (of parameter 1). This matches previously obtained results about the \emph{induced measure} of mixed quantum states (density matrices) \cite{zyczkowski2001induced,sommers2004statistical,nechita2007asymptotics}. Indeed, the matrix $\rho_A$ can be interpreted in quantum information theory as the partial trace of the rank-one matrix in the direction of a random Gaussian vector $\Psi_G \in \C^d \otimes \C^d$. Up to normalization, this random density matrix belongs to the ensemble of induced density matrices. The fact that the two factors of the tensor product have equal dimensions corresponds to taking the \emph{uniform measure} on the (convex, compact) set of density matrices \cite{zyczkowski2003hilbert, zyczkowski2011generating}. The statistics of the eigenvalues of such random matrices have been extensively studied in the literature. In particular, the asymptotic von Neumann entropy has been studied by Page \cite{page1993average,foong1994proof,sanchez-ruiz1995simple}, who conjectured that
$$\E S(\rho_A)= \sum_{i=D+1}^{2D} \frac 1 i - \frac{D-1}{2D} \sim \log D - \frac 1 2 \quad \text{ as } D \to \infty.$$
We refer to \cref{sec:Results for normalized tensor network states} for a derivation of such statistics in the context of our work.

\subsection{Series network}

Let us now consider a tensor network consisting of $s$ vertices arranged in a path graph, with two half-edges at the end points. We depict this network, as well as the various steps needed to compute the limiting spectrum distribution of the reduced matrix. The network associated to the graph (where the partition of the half-edges is clear) has a single path from the source to the sink, so the maximum flow is unity. 

\begin{figure}[htb]
    \centering
    \includegraphics{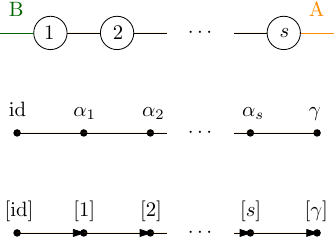}\bigskip
    \caption{A series network. The $s$ vertices of the network are arranged in a line, with two half-edges at the end points. The maximum flow is 1, and it is unique. The single path realizing the flow induces a (total) order $\alpha_1 \leq \alpha_2 \leq \cdots \leq \alpha_s$ on the geodesic permutations.}
    \label{fig:series-network}
\end{figure}

The residual graph, obtained by removing the edges from the unique path achieving maximum flow, is empty. Hence, the partial order on the vertices is again a total order: 
$$\id \preceq \alpha_1 \preceq \cdots \preceq \alpha_s \preceq \gamma.$$

We have thus a \emph{series network}, and the final measure can be obtained by applying $s$ times the series concatenation procedure from \cref{def:series-parallel-probability-measure} to obtain 
$$\mu_{G_{A|B}^o} = \underbrace{\MP \boxtimes \MP \boxtimes \cdots \boxtimes \MP}_{s \text{ times}} = \MP^{\boxtimes s}.$$

Let us note that very similar results were previously obtained by Cécilia Lancien \cite{lancien}, see also \cite{cheng2024random}. This measure is commonly know as the \emph{Fuss-Catalan} distribution of order $s$ \cite{banica2011free}, see also \cref{th: Mp law and FC}. Its moments are known in combinatorics as the Fuss-Catalan numbers: 
$$\int t^n \, \mathrm{d} \MP^{\boxtimes s}(t) = \frac{1}{sn + 1} \binom{sn + n}{n}$$
and its entropy is \cite[Proposition 6.2]{collins2010randoma}
$$\int - t \log t  \, \mathrm{d} \MP^{\boxtimes s}(t) = \sum_{i=2}^{s+1} \frac 1 i.$$
Such tensor network states have already been considered in quantum information theory \cite{collins2010randoma,collins2013area,zyczkowski2011generating}

\subsection{2D lattice}

We now discuss a physically relevant network: a rectangle that is part of a $2D$ lattice (part of $\mathbb Z^2$). We have thus two integer parameters, the length $L$ and the height $H$ of the rectangle, and $H \cdot L$ vertices. The vertices are connected by the edges inherited from the $\mathbb Z^2$ lattice, see \cref{fig:2D-lattice-network}. The left-most (resp.~right-most) columns of vertices have half-edges that belong to the class $B$ (resp.~A) of the half-edge partition defining the two regions. 

\begin{figure}[htb]
    \centering
    \includegraphics{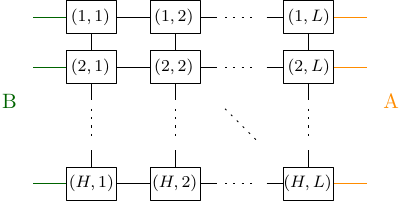}\bigskip
    \caption{A network corresponding to a $H \times L$ 2D lattice. Half-edges are attached to vertices on the left and right boundary, corresponding to a tensor in $(\mathbb C^D)^{\otimes 2H}$.}
    \label{fig:2D-lattice-network}
\end{figure}

The flow network corresponding to the graph and the partition $A|B$ is depicted in \cref{fig:2D-lattice-flow-order}, top diagram. The maximum flow in this network is $H$: one can consider $H$ parallel horizontal paths which go from $\id$ to $\gamma$. Note that the set of $H$ edge-disjoint paths in the network achieving the maximum flow is unique. The residual network is non-empty in this case, with $H$ clusters of the form 
$$ \mathcal C_j := \{[i,j] \, : \, i =1, \ldots, L\}.$$ 
The order relation on the clusters is again a total order on $L$ points, see \cref{fig:2D-lattice-flow-order}, bottom diagram. We are thus recovering again the Fuss-Catalan distribution: 
$$\mu_{G_{A|B}^o} = \MP{^\boxtimes L}.$$

\begin{figure}[htb]
    \centering
    \includegraphics{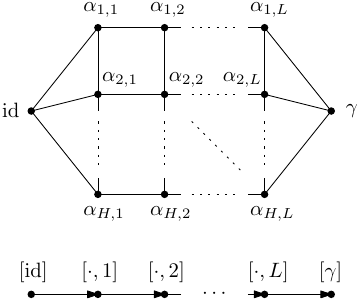}\bigskip
    \caption{The newtwork associated to the $H \times L$ $\mathbb Z^2$ lattice fragment. The maximum flow in this network is $H$, corresponding to horizontal edge-disjoint paths. These $H$ paths induce a total order on the $L$ vertex clusters, each having $H$ vertices.}
    \label{fig:2D-lattice-flow-order}
\end{figure}

\section{Results for normalized tensor network states}\label{sec:Results for normalized tensor network states} 
In this section, we will give our main technical contribution. With the help of all the results obtained from the previous sections, we will be able in this section to compute the R\'enyi and von Neumann entropy for a given approximated normalised state $\tilde{\rho}_A:=D^{-|E_{\partial}|}\rho_A$ associated to a given boundary subregion $A\subseteq E_{\partial}$. The main results of this section consist first on showing the weak convergence of moments associated to an approximated reduced state $\tilde{\rho}_A$ associated with a given boundary region $A$ in Theorem \ref{th:weack convergence}. Moreover we will show in Corollary \ref{corr: average entropie scaling} the existence of correction terms as moments of a graph-dependent measure which can be explicitly computed in the case of an obtained series-parallel partial order $G_{A|B}^o$. 

In Subsection \ref{subsec: concentration} we will show different concentration inequalities, which will allows us in Subsection \ref{subsec: entnaglement} to give the main results of this section.
\subsection{Concentration}\label{subsec: concentration}
In this subsection, we will give different concentration results that will allows us in the following subsection to give our main technical contribution. 

First, we recall the following theorem that estimates the deviation probability of polynomials in Gaussian random variables. This theorem will be relevant for different concentration results that we will proof in the rest of this subsection.
\begin{theorem}\label{Th: concentration poly gaussian}
    Let $g$ be a polynomial in $m$ variables of degree $q$. Then, if $G_1,\cdots,G_m$ are independent centered Gaussian variables, 
    \begin{equation*}
        \forall t>0,\quad \mathbb P\left(\big|g(G_1,\cdots,G_m)-\E \,g\big|>t(\operatorname{Var}(g))^{\frac{1}{2}}\right)\leq \exp{\left(-c_q\,t^{\frac{2}{q}}\right)},
    \end{equation*}
    where $V(g)$ is the variance of $g(G_1,\cdots,G_m)$ and $c_q$ is a constant which depends only on $q$.
\end{theorem}
\begin{proposition}\label{prop: concentration trace}
Let $G$ a bulk connected graph and let $A\subseteq E_{\partial}$ then: 
\begin{equation*}
    \mathbb P\Big(\Big|\Tr \tilde{\rho}_A-1\Big|>\epsilon\Big)\leq\exp{\left(-c_{|E|}\epsilon^{\frac{1}{|E|}}D^{\frac{|E_b|}{2|E|}}\right)}, 
\end{equation*}
where $\tilde{\rho}_A:=D^{-|E_{\partial}|}\rho_A$.
\end{proposition}
\begin{proof}
First remark that $\Tr \rho_A$ is a $2|E|$ polynomial in $\ket{g_x}\in \mathcal{H}_x$, moreover we recall that for random Gaussian vector $\ket{g_x}\in\mathcal{H}_x$ one have: 
\begin{equation*}
    \forall x\in V,\quad \E\left[\ketbra{g_x}{g_x}\right]=\id_x\quad\text{and}\quad \E\left[\ketbra{g_x}{g_x}^{\otimes 2}\right]=\id_x+F_x,
\end{equation*}
where the $\id_x$ and $F_x$ acts in all the edges of Hilbert space generating the local Hilbert space for each vertex $x$. Moreover, it is implicitly assumed that $\id_x\equiv \id_x^{\otimes 2}$ and the Swap operator $F_x$ is a unitary representation of permutation element in $\mathcal{S}_2$.

It is easy to check the variance $\operatorname{Var}(\Tr\tilde{\rho}_A)$ gives: 
\begin{equation*}
    \operatorname{Var}(\Tr\tilde{\rho}_A)=\E\left[(\Tr\tilde{\rho}_A)^2\right]-\left(\E\left[\Tr(\tilde{\rho}_A)\right]\right)^2=O\left(D^{-|E_b|}\right), 
\end{equation*}
where we have used that: 

\begin{align*}
\E\left[(\Tr\tilde{\rho}_A)^2\right]&=\Tr\left(\bigotimes_{e\in E_b}\ketbra{\Omega_e}{\Omega_e}^{\otimes 2}\bigotimes_{x\in V}\E\left[\ketbra{g_x}{g_x}^{\otimes 2}\right]\right)\\
&=1+D^{|E_{\partial}|}\prod_{e\in E_b}\Tr\left(\ketbra{\Omega_e}{\Omega_e}^{\otimes 2}\,F_e\right)\prod_{e\in E_{\partial}}\Tr\left(F_e\right)\\
&=1+O\left(D^{-|E_b|}\right),
\end{align*}
where in the last equality the bulk contribution is of $D^{-|E_b|}$ while the boundary edges contribute with $D^{|E_{\partial}|}$.
The second term of the variance is :
\begin{equation*}
\E\left[\Tr(\tilde{\rho}_A)\right]=\Tr\left(\bigotimes_{e\in E_b}\ketbra{\Omega_e}{\Omega_e}\bigotimes_{x\in V}\E\left[\ketbra{g_x}{g_x}\right]\right)=1.
\end{equation*}
By combining the variance  $\operatorname{Var}(\Tr(\tilde{\rho}_A))$ with Proposition \ref{Th: concentration poly gaussian} one have: 
\begin{equation*}
    \mathbb P\Big(\big|\Tr \tilde{\rho}_A-\E \Tr \tilde{\rho}_A\big|>\epsilon\Big)\leq\exp{\left(-c_{|E|}\epsilon^{\frac{1}{|E|}}D^{\frac{|E_b|}{2|E|}}\right)}, 
\end{equation*}
where we have defined $\epsilon:=t\left(D^{-\frac{|E_b|}{2}}\right)$ with $c_{|E|}> 0$ is a constant depending only in the total number of edges $|E|$.

\end{proof}
\begin{proposition}\label{prop:concentration moments}
Let $G$ a bulk connected graph and let $A\subseteq E_{\partial}$ we have:
    \begin{equation*}
\forall n>1,\quad \mathbb P\Big(\Big|\frac{1}{D^{F(G_{A|B})}}\Tr(\sigma_A^n)-\frac{1}{D^{F(G_{A|B})}}\E[\Tr(\sigma_A^n)]\Big|>\epsilon\Big)\leq\exp{\Big(-c_{2n|E|}D^{\frac{1}{2n|E|}}\epsilon^{\frac{1}{n|E|}}\Big)},
\end{equation*}

    where $\sigma_A:=D^{F(G_{A|B})}\tilde{\rho}_A$.
\end{proposition}
\begin{proof}
  The proof of this proposition follows the same proof spirit of the proposition above. Remark that $\Tr\sigma_A^n$ is a $2n|E|$ polynomial in $\ket{g_x}$. Moreover the variance was estimated in \cite[Lemma 14]{hastings2017asymptotics} where:
  \begin{equation*}
      \operatorname{Var}\left(\frac{1}{D^{F(G_{A|B})}}\Tr(\sigma_A^n)\right)=O\left(\frac{1}{D}\right).
  \end{equation*}
  By defining $\epsilon:=tD^{-\frac{1}{2}}$ we obtain the desired result.
\end{proof}
\subsection{Entanglement entropy}\label{subsec: entnaglement} In this subsection we will introduce the main technical contribution of this work. With the help of concentration results, we will first assume and work with the approximate normalised state $\tilde{\rho}_A:=D^{-|E_{\partial}|}\rho_A$. 
We will show that as $D\to\infty$ one can compute the average R\'enyi and von Neumann entanglement entropy with correction terms. In particular if the obtained partial order is series-parallel, the correction terms will be given as moment of an partial order dependent measure $\mu_{G_{A|B}^o}$.

We recall first from Subsection \ref{subsec: entanglement of a subregion} that the rank of the approximate normalised state is upper bounded by $D^{F(G_{A|B})}$. Let consider the restricted approximate normalised quantum state $\tilde{\rho}_A$ to its support and its empirical measure $\mu^{(D)}_A$ defined as: 
\begin{equation*}
    \sigma_A:=D^{F(G_{A|B})}\tilde{\rho}_A^S,\quad\text{and}\quad\mu_A^{(D)}:=\frac{1}{D^{F(G_{A|B})}}\sum_{\lambda\in \mathrm{spec}(\sigma_A)}\delta_{\lambda},
\end{equation*}
where $\tilde{\rho}_A^S$ is the reduced approximate normalised state restricted on its support. The definition of $\tilde{\rho}_A$ and the empirical measure $\mu_A^{(D)}$ will allow us to show in Theorem \ref{th:weack convergence} the weak convergence of $\mu_{A}^{(D)}$ to $\mu_{G_{A|B}}$. In particular if the obtained partial order $G_{A|B}^o$ is series-parallel from Theorem \ref{Th: moments and free pro} one will have weak convergence to $\mu_{G_{A|B}}$. This result will allow us in Corollary \ref{corr: average entropie scaling} to compute the R\'enyi and von Neumann entanglement entropy.

 Recall first, that a measure $\mu^{(D)}$ \emph{converges weakly} to a measure $\mu$ if for any continuous function $f:\R\to\R$ we have: 
\begin{equation*}
    \forall \epsilon >0,\quad \lim_{D\to\infty}\mathbb P\left(\left|\int f(t)\mathrm{d}\mu^{(D)}(t)- \int f(t)\mathrm{d}\mu(t)\right|\leq \epsilon \right)=1.
\end{equation*}
\begin{theorem}\label{th:weack convergence}
Let boundary region $A\subseteq E_{\partial}$ in the graph $G$. The empirical measure $\mu_A^{(D)}$ associated to the approximated normalised state $\sigma_A$ converges weakly to $\mu_{G_{A|B}}$. More precisely for all continuous function $f:\R\to\R$ we have:
 \begin{equation*}
    \forall \epsilon >0,\quad \lim_{D\to\infty}\mathbb P\left(\left|\int f(t)\mathrm{d}\mu^{(D)}_A(t)- \int f(t)\mathrm{d}\mu_{G_{A|B}}(t)\right|\leq \epsilon \right)=1.
\end{equation*}
\end{theorem}
\begin{proof}
As was shown in Theorem \ref{thm:limit-moments-general} the moment converges to a unique measure $\mu_{G_{A|B}}$. In the particular case of an ordered series-parallel graph $G_{A|B}^o$ we have an explicit graph dependent measure $\mu_{G_{A|B}^o}$. Recall from Theorem \ref{thm:limit-moments-general} that: 
\begin{equation*}
    \frac{1}{D^{F(G_{A|B})}}\E\Big[\Tr\left(\sigma_A^n\right)\Big]\xrightarrow[D \to \infty]{}m_{n}=\int t^n \mathrm{d}\mu_{G_{A|B}}(t).
\end{equation*}
From standard probability theory results the convergence in probability implies weak convergence (see \cite[Theorem 25.2]{billingsley2012probability}. For that one needs only to show the decreasing scaling of the variance as $D\to\infty$. By using \cite[Lemma 14]{hastings2017asymptotics} that: 
\begin{equation*}
    \operatorname{Var}\left(\frac{1}{D^{F(G_{A|B})}}\Tr(\sigma_A^n)\right)=O\left(\frac{1}{D}\right),\quad(D\to\infty),
\end{equation*}
hence the weak convergence of $\mu_A^{(D)}$ to $\mu_{G_{A|B}}$, in particular if the graph is series-parallel we have $\mu_{G_{A|B}^o}$.
\end{proof}

\begin{lemma}\label{lemma: concentration of log}
Let boundary region $A\subseteq E_{\partial}$ and let $m_n^{(D)}$ the moment associated to the empirical measure $\mu_A^{(D)}$ one have:
    \begin{equation*}
        \mathbb P\Big(\Big|\E\log\Big(m_{n}^{(D)}\Big)-\log\Big(\E m_{n}^{(D)}\Big)\Big|>\epsilon\Big)\xrightarrow[D \to \infty]{}1\quad\text{where}\quad m_n^{(D)}:=\frac{1}{D^{F(G_{A|B})}}\E\Big[\Tr\left(\sigma_A^n\right)\Big].
    \end{equation*}
\end{lemma}
\begin{proof}
By Proposition \ref{prop:concentration moments} and Jensen's inequality that $\E\log\Big(m_{n}^{(D)}\Big)\leq\log\Big(\E m_{n}^{(D)}\Big)$. All what remains to show that $\E\log\Big(m_{n}^{(D)}\Big)\geq\log\Big(\E m_{n}^{(D)}\Big)$ holds with high probability. Fix $\epsilon>0$. From Proposition \ref{prop:concentration moments} we know that 
    \begin{equation*}
        m_n^{(D)}\geq \E m_n^{(D)}-\delta\quad\text{with}\quad 0<\delta\leq\frac{\epsilon}{\epsilon+1}\E m_n^{(D)},
    \end{equation*}
holds with probability $1-\exp{\Big(-c_{2n|E|}D^{\frac{1}{2n|E|}}\delta^{\frac{1}{n|E|}}\Big)}$. It is easy to check that the following inequalities hold: 
\begin{align*}
    \log\left(m_n^{(D)}\right)\geq \log\left(\E m_n^{(D)}-\delta\right)&=\log\left(\E m_n^{(D)}\right)+\log\left(1-\frac{\delta}{\E m_n^{(D)}}\right)\\
    &\geq \log\left( \E m_n^{(D)}\right)-\frac{\delta}{\E m_n^{(D)}-\delta}\geq \log\left(\E m_n^{(D)}\right)-\epsilon.
\end{align*}
Therefore we have that $\E\log\left(m_n^{(D)}\right)\geq \log\left(\E m_n^{(D)}\right)-\epsilon$ occurs with probability at least
\begin{equation*}
    1-\exp{\Big(-c_{2n|E|}D^{\frac{1}{2n|E|}}\delta_{\max}^{\frac{1}{n|E|}}\Big)}\quad\text{where}\quad \delta_{\max}=\frac{\epsilon}{\epsilon+1}\E m_n^{(D)}.
\end{equation*}
As $D\to \infty$, $\E m_n^{(D)}$ converges, hence $\delta_{\max}=O(1)$, showing that the probability estimate above converges to 1 and finishing the proof.
\end{proof}
We recall for completeness the following proposition from \cite{collins2011gaussianization} which will play a key role for the proof of our main result. 
\begin{proposition}\cite[Proposition 4.4]{collins2011gaussianization}\label{prop: proposition 4.4 Ion Collins}
Let $f$ be a continuous function on $\R$ with polynomial growth and $\nu_n$ a sequence of probability measures which converges in moments to a compactly supported measure $\nu$. Then $\int f\mathrm{d}\nu_n\to\int f \mathrm{d}\nu$.
    
\end{proposition}
\begin{corollary}\label{corr: average entropie scaling}
Let boundary region $A\subseteq E_{\partial}$ in $G$, and let $\tilde{\rho}_A$ the approximated reduced normalised state. Then the averaged R\'enyi and von Neumann entropy converges weakly as $D\to \infty$ are given by:
\begin{align*}
        F(G_{A|B}) \log D -\E S_n(\tilde{\rho}_A) &\xrightarrow[D \to \infty]{} \frac{1}{n-1}\log\left(\int t^n\,\mathrm{d}\mu_{G_{A|B}}\right),\\
        F(G_{A|B}) \log D- \E S(\tilde{\rho}_A)&\xrightarrow[D \to \infty]{} \int t\,\log t\,\mathrm{d}\mu_{G_{A|B}}.
\end{align*}
where $F(G_{A|B}):=\operatorname{maxflow}(G_{A|B})$ .    
\end{corollary}
\begin{proof}
The poof of this corollary is a direct consequence of different obtained concentration results from the previous subsection and the weak convergence of $\mu^{(D)}_A$ to $\mu_{G_{A|B}}$. 

First, we shall start with the R\'enyi entropy, for that let consider: 
\begin{align*}
    F(G_{A|B}) \log D -\E S_n(\tilde{\rho}_A)=\frac{1}{1-n}\E \log\left(m_{n,A}^{(D)}\right),\quad\text{where}\quad m_n^{(D)}:=\frac{1}{D^{F(G_{A|B})}}\E\Big[\Tr\left(\left(\sigma_A\right)^n\right)\Big],
\end{align*}
and recall that $\sigma_A:=D^{F(G_{A|B})}\tilde{\rho}_A^S$ restricted on the support of $\tilde{\rho}_A:=D^{-|E_{\partial}|}\rho_A$. By using Lemma \ref{lemma: concentration of log} and in the limit $D\to\infty$ we have: 
\begin{equation*}
F(G_{A|B}) \log D -\E S_n(\tilde{\rho}_A)\xrightarrow[D \to \infty]{} \frac{1}{n-1}\log\left(\int t^n\,\mathrm{d}\mu_{G_{A|B}}\right).
\end{equation*}
For the von Neumann entropy let consider $\{\lambda_i\}\in\mathrm{spec}(\sigma_A)$ and $\{\tilde{\lambda}_i\}\in\mathrm{spec}(\tilde{\rho}_A)$, it is direct that: 
\begin{equation*}
    \E S(\tilde{\rho}_A)=-\E\sum_i\tilde{\lambda}_i\log(\tilde{\lambda}_i)=F(G_{A|B})\log(D)-\frac{1}{D^{F(G_{A|B})}}\E\sum_i\lambda_i\log(\lambda_i).
\end{equation*}
Define the function $f:\R\to\R$ as $f(t):=t\log t$, by combining Proposition \ref{prop: proposition 4.4 Ion Collins} and Theorem \ref{th:weack convergence} we have the following weak convergence as $D\to\infty$
\begin{equation*}
    F(G_{A|B})\log D- \E S(\hat \rho_A)=\frac{1}{D^{F(G_{A|B})}}\,\E\left(\sum_i\,f(\lambda_i)\right)\xrightarrow[D \to \infty]{}\int\,f(t)\,\mathrm{d}\mu_{G_{A|B}},
\end{equation*}
where the measure $\mu_{G_{A|B}}$ is defined on a compact support, ending the proof of the corollary. In the particular case if the obtained poset structure $G_{A|B}^o$ is series parallel the obtained graph dependent measure is explicitly given $\mu_{G_{A|B}}=\mu_{G_{A|B}^o}$ by Theorem \ref{Th: moments and free pro}. 
\end{proof}

\section{Conclusion}
From a given graph general graph with boundary region and bulk region, the main goal of this work is to compute the entanglement entropy, the R\'enyi and the von Neumann entropy, of a given sub-boundary region $A$ of the graph. By analysing as $D\to \infty$ the moments of a state associated to the region $A$, with the help of the (maximal) flow approach we computed the leading terms contribution to the moment. By analysing and removing all the augmenting paths starting from $\id$ and ending in $\gamma$ 
 of the network $G_{A|B}$ constructed by connecting the region $A$ to the total cycle $\gamma$ and $\id$ to the region $B$ one obtains a cluster graph $G_{A|B}^c$ by identifying all the remaining edges connected permutations. The flow approach induces a natural ordering poset structure  represented by the induced poset order $G_{A|B}^o$. The maximal flow approach allows us to deduce the moment convergence to the moment of a unique graph-dependent measure $\mu_{G_{A|B}}$. This result allows us to deduce the higher order correction terms of the Rényi and von Neumann entropy given by a graph-dependent measure $\mu_{G_{A|B}}$. Moreover, we have shown if the obtained partial order $G_{A|B}^o$ is series-parallel, and with the hep of free probability theory we can explicitly give the associated graph-dependent measure $\mu_{G_{A|B}}=\mu_{G_{A|B}^o}$ that will contribute to the higher order correction terms of each of the Rényi and von Neumann entanglement entropy.   
 
In this work, we did not assume any assumption on the minimal cuts, in the maximal flow approach by duality one can obtain different minimal cuts which may intersect in different edges. Moreover, the higher-order correction terms in the entanglement entropy can describe the quantum corrections beyond the area law behaviour of the expected Ryu-Takayanagi entanglement entropy in the context of ADS/CFT. It was previously argued in the literature that if one wants to consider higher-order correction terms in the random tensor network setting one needs to go beyond the maximally entangled state and consider general link states representing the bulk matter field. In this work the obtained higher-order quantum fluctuation of entanglement entropy with only maximally entangled states that we interpret as fluctuations of spacetime itself without any need of bulk fields represented by a generic link state.    

\bigskip

\noindent\textbf{Acknowledgments.} We would like to thank Cécilia Lancien for sharing with us preliminary notes on very similar questions. The authors were supported by the ANR projects \href{https://esquisses.math.cnrs.fr/}{ESQuisses}, grant number ANR-20-CE47-0014-01, and \href{https://www.math.univ-toulouse.fr/~gcebron/STARS.php}{STARS}, grant number ANR-20-CE40-0008, as well as by the PHC program \emph{Star} (Applications of random matrix theory and abstract harmonic analysis to quantum information theory). K.F.~acknowledges support from a \href{https://nanox-toulouse.fr/}{NanoX} project grant.

\bibliographystyle{alpha}
\bibliography{refs}

\appendix 
\section{Basics of the combinatorial approach to free probability theory}\label{sec: appendix}
In this section, we will recall and give the necessary material on \emph{combinatorics} and \emph{free probability} needed to understand the rest of this section. All the material that we shall introduce is standard and can be found in \cite{nica2006lectures,mingo2017free}. 

Let $\pi:=\{V_1,\cdots,V_n\}$\footnote{Do not confuse with $\pi_i$ introduced in Section \ref{sec: Optimisation with the MF} representing the different paths.} be a partition of a finite totally ordered set $S$ such that $\bigsqcup_{i\in[n]}V_i=S$. We call $\{V_i\}$ the \emph{blocks} of $\pi$. We denote by $p\sim_{\pi}q$ if $p$ and $q$ belongs to the same block of $\pi$. A partition $\pi$ of a set $S$ is called \emph{crossing} if there exists $p_1<q_1<p_2<q_2$ in $S$ such that $p_1\sim_{\pi}p_2\nsim_{\pi}q_1\sim q_2$. We called a \emph{non-crossing} partition if $\pi$ is not crossing. We note by $\NC(S)$ the non-crossing partition set of $S$. In particular if $S=\{1,\cdots n\}$, we denote the non-crossing partition by $\NC(n)$. The set of non-crossing partition plays a crucial in different areas from combinatorics \cite{armstrong2007generalized} to random matrices and free probability theory which will be our main focus. Moreover one should mention a crucial result \cite{nica2006lectures}: \emph{there exists a one-to-one correspondence of the non-crossing partition set and the set of permutations $\alpha$ in a geodesic between $\gamma$ and $\id$ i.e $|\alpha|+|\alpha^{-1}\gamma|=|\gamma|$.} Another important fact, the cardinality $|\NC(n)|=\mathrm{Cat}_n$ where:
\begin{equation}\label{eq: Catalan numbers}
    \mathrm{Cat}_n:=\frac{1}{n+1}\binom{2n}{n},
\end{equation}
are the \emph{Catalan numbers}. For more combinatorial details and properties of the Catalan numbers and the non-crossing partitions see \cite{nica2006lectures}. Assume $(\alpha_1,\cdots,\alpha_k)$ $k$ tuples of permutations in $\mathcal{S}_n$ such that  \begin{equation}\label{eq: geo equation}
|\alpha_1|+\sum_{i\in[k-1]}|\alpha_{i}^{-1}\alpha_{i+1}|+|\alpha_k^{-1}\gamma|=|\gamma|.
\end{equation}
are geodesics between $\id$ and $\gamma$. The cardinality of the set of the $k$ tuple permutations satisfying the geodesic equation \eqref{eq: geo equation} known as the \emph{Fuss-Catalan} numbers given by: 
\begin{equation*}
    \mathrm{FC}_{n,k}:=\frac{1}{nk+1}\binom{n+nk}{n}, 
\end{equation*}
generalizing the Catalan numbers for $k=1$.

Now we are ready to introduce the \emph{free probability theory} tools that will be used in this work. Moreover, one should mention the intrinsic link between free probability theory and combinatorics where we will give some examples to illustrate it. The combinatorics will allow us in the rest of this section to understand our main result.

We recall that a \emph{non-commutative probability space} is a pair $(\mathcal A,\omega)$ of a unital $C^*$-algebra $\mathcal A$ with a state \emph{state} $\omega:\mathcal{A}\to \C$ such that $\omega(1_{\mathcal{A}})=1$. One says that the elements $a\in\mathcal A$ define a \emph{noncommutative variable}. In the non-commutative probability space, one can associate the distribution law $\mu_a$ to $a\in \mathcal A$ which is defined as $\mu_a=\omega(a)$. 

Before we give some concrete examples of some non-commutative probability spaces, we shall recall the notion of \emph{freeness} that plays a crucial role in the non-commutative probability world. The notion of freeness generalizes the ``classical" independence when the algebra $\mathcal A$ is commutative. We say that for a given $n$ non-commutative random variables $\{a_i\}\in\mathcal A$ are \emph{free independent} if for any polynomials $\{p_i\}$ the following holds:  
\begin{equation}
    \omega(a_1a_2\cdots a_n)=0
\end{equation}
whenever $\omega(p_k(a_{i_k}))=0$ for $k\in[n]$ and two no adjacent indices $i_k$ and $i_{k+1}$. One can check that with the definition of free independence one has for given two free independent variables $a_1$ and $a_2$: 
\begin{equation}
   \omega((a_1-\omega(a_1))(a_2-\omega(a_2))=\omega(a_1a_2)-\omega(a_1)\omega(a_2)=0,
\end{equation}
hence, generalizing the notion of standard independence in the commutative setting where $\mathbb E(a_1a_2)=\mathbb E(a_1)\mathbb E(a_2)$ for two commutative random variables $a_1,a_2$ in a commutative probability space. 
\begin{definition}
    Let $(\mathcal{A}_N,\omega_N)$ with $N\in\mathbb N$ and $(\mathcal{A},\omega)$ non-commutative probability spaces. We say that $a_N\in\mathcal{A}_N$ converges weakly to $a\in\mathcal{A}$ as $N\to\infty$ if the following holds: 
    \begin{equation}
        \lim_{N\to\infty}\omega_N((a_N)^n)=\omega(a^n)\quad \forall n\in\mathbb N,
    \end{equation}
    where $\omega(a^n)=\int x^n \mathrm d\mu_a(x)$ are the moments of $a$.
\end{definition}
To illustrate concrete non-commutative probability spaces, we give some classical examples. The first example we shall deal with is ``classical" probability space corresponding to commutative algebra. For that let $(\Omega,\Sigma,\mu)$ where $\Omega$ a set, $\Sigma$ a $\sigma-$algebra, and $\mu$ probability measure. Define $\mathcal{A}:=L^{\infty-}(\Omega,\mu)$ where:
\begin{equation*}
L^{\infty-}(\Omega,\mu):=\bigcap_{1\leq k<\infty}L^{k}(\Omega,\mu),    
\end{equation*}
and the state $\omega$ as:\begin{equation*}
\omega(a):=\int_{\Omega}a(x)\mathrm d\mu(x),\quad a\in\mathcal{A}.
\end{equation*} 
 The tuple $(\mathcal{A},\omega)$ defines a commutative probability space. Another standard example that can be considered is the random matrices case. Let us consider the algebra $\mathcal A$ consisting of valued $k\times k$ matrices over $L^{\infty-}(\Omega,\mu)$ where \begin{equation*}
    \mathcal{A}:=\mathcal{M}_k(L^{\infty-}(\Omega,\mu)). 
\end{equation*} 
Define the state $\omega$ on $\mathcal{A}$ as: 
\begin{equation*}
    \omega(a):=\int_{\Omega} tr(a(x))\mathrm d\mu(x),\quad a\in\mathcal{A},
\end{equation*}
where $tr(\cdot)$ is the normalized trace. The space $(\mathcal{A},\omega)$ define a non-commutative probability space which the space of random matrices over $(\Omega,\Sigma,\mu)$. 

 We recall for a given non-commutative random variable $a\in\mathcal A$, the $n$th moments of $a$ are given by 
\begin{equation}
    m_{n}(\mu_a):=\int t^n \mathrm d\mu_a(t).
\end{equation}
Moreover for a given random variables $\{a_1,\cdots,a_n\}$ in $\mathcal{A}$, the moments are given by 
\begin{equation}
\omega(a_1\cdots a_n):=\sum_{\pi\in \NC(n)}\kappa_{\pi}(a_1,\cdots,a_n),
\end{equation}
where $\kappa_{\pi}$ are the \emph{free cummulants}. The equation given above is known as the \emph{moments-cummulants formula}, where the free independence can be characterized by the vanishing of mixed cumulants (see \cite[Chapter 11]{nica2006lectures}).

In free probability theory, for two free independent random variables $a_1,a_2\in \mathcal{A}$, one can define a ``convolution operation".
Mostly in this work, we only shall deal with the \emph{free multiplicative convolution}. Let $a_1$ and $a_2$, two free independent random variables in $\mathcal A$ with their respective distribution $\mu_{a_1}$ and $\mu_{a_2}$. A free multiplicative convolution or simply a \emph{free product} is defined by 
\begin{equation}
    \mu_{a_1a_2}:=\mu_{a_1}\boxtimes\mu_{a_2},
\end{equation}
where $\mu_{a_1a_2}$ represents the distribution of $a_1\,a_2$. There exists a standard and analytical way to compute the free product, like for the free additive convolution, which can be done via the \emph{S-transform}.
The \emph{S-transform} of a probability distribution $\mu_a$ is defined as: 
\begin{equation}
    S_{\mu_a}(z):=\int \frac{1}{x-z}\mathrm d\mu_a(x),
\end{equation}
which is analogous to the R-transform for the free additive convolution, as we shall describe. 
Moreover, it can also computed equivalently by \emph{the formal inverse of the moment-generating formal power series} given by:
\begin{equation}
    S_{\mu_a}(z)=\frac{1-z}{z}M_a^{-1}(z),
\end{equation}
where $M_a^{-1}(z)$ is the formal inverse of the \emph{moment-generating} formal power series given by 
\begin{equation}
    M_{\mu_a}(z):=\sum_{k=1}^{\infty}\,m_{k,a}\,z^k.
\end{equation}
With the help of the S-transform, one can compute the free product where: 
\begin{equation}
    S_{\mu_{a_1a_2}}(z)=S_{\mu_{a_1}}(z)\,S_{\mu_{a_2}}(z)=S_{\mu_{a_1\boxtimes \mu_{a_2}}}(z).
\end{equation}
In the following, we shall recall some standard distributions that are well-known in the literature and will be highly used in this work. 

    The first distribution we shall consider here is the \emph{semicircular law}, it is one of the most important distributions we encounter in free probability theory. The semicircular distribution $\mu_{\text{SC}}(x)$ is defined by the density: 
    \begin{equation}
        \mathrm{d}\mu_{\text{SC}}(x):=\frac{\sqrt{4-x^2}}{2\pi}\mathbbm 1_{x\in[-2,2]}\mathrm dx.
    \end{equation}
For an illustration, and by standard computation, one can compute the S-transform of the semi-circular distribution: 
\begin{equation*}
    S_{\mu_{\text{SC}}}(z)=\frac{-z+\sqrt{z^2-4}}{2}.
\end{equation*}
    The first link that can be made, is the moments of the semicircular distribution are intrinsically related to the Catalan numbers. One can easily check that the following equality holds:
    \begin{equation}
        \int x^{k}\mathrm d\mu_{\text{SC}}(x)=\mathrm{Cat}_k, 
    \end{equation}
    where $\mathrm{Cat}_k$ are the Catalan numbers see equation \eqref{eq: Catalan numbers}, and Chapter 2 in \cite{nica2006lectures} for more details. Moreover, one should say that the S-transform gives another important link between free probability and combinatorics by computing the moments of free product convolution of Mar\u{c}henko-Pastur distribution (see Theorem \ref{th: Mp law and FC}).
    
    Another well-known, due to Wigner \cite{wigner1993characteristic} shows the following result linking the semicircular distribution and random Gaussian matrices.
    \begin{theorem}\cite{wigner1993characteristic}
        Let $N\in\mathbb N$, let $A_N$ be an $N\times N$ selfadjoint random Gaussian matrix. Then $A_N$ converges weakly to a semicircular distribution $\mu_{SC}(x)$.
    \end{theorem}
We refer to \cite{anderson2010introduction} for a complete proof. As we have shown in this particular case the existence of a deep link between random Gaussian matrices, the semicircular law, and the Catalan numbers. Moreover, the semicircular law plays an important role in free probability theory as a free central limit distribution. We recall one of the main results in free probability theory, see Theorem 8.10 in \cite{nica2006lectures}. 
    \begin{theorem}
        Let $(\mathcal{A},\omega)$ a non-commutative probability space and $a_1,\cdots,a_N\in\mathcal{A}$ free independent and identically distributed self-adjoint random variables. Assuming that $\omega(a_i)=0$ for $i\in[N]$ and denote by $\sigma^2:=\omega(a_i^2)$ the variance of the random variables $a_i$. Then the following holds: 
        \begin{equation*}
            \frac{a_1+\cdots+a_N}{\sqrt{N}}\to \mu_{SC}(x). 
        \end{equation*}
        converges weakly to $\mu_{SC}(x)$ as $N\to\infty$. Where $s$ is a semicircular of variance $\sigma^2$.
    \end{theorem}
With this particular distribution, we have shown how random matrices, combinatorics, and free probability theory can be related. 

In the following, we will give another example of distribution that will play an important role in this work. The second distribution we shall consider is the \emph{Mar\u{c}henko-Pastur distribution}. We shall denote by $\MP(t)$ defined by:
\begin{align*}
    \MP(t)&:=\max(1-t,0)\delta_0+\nu_t,\\
    d\nu_t(x)&:=\frac{\sqrt{4t-(x-1-t)^2}}{2\pi x}\mathbbm 1_{(x-1-t)^2\leq 4t}\mathrm dx.
\end{align*}
Recall the Mar\u{c}henko-Pastur distribution is deeply related to Wishart matrices. Let $Z$ a Whishart matrix defined $Z:=\frac{1}{m} YY^*$, where $Z\in\M_{nm}(\C)$ where the entries of $Y\in\M_{nm}(\C)$ are complex random Gaussian variables. It was shown by Mar\u{c}henko and Pastur that the empirical distribution of Whishart matrices converges to the $\MP(t)$ defined above. More precisely they have shown the following theorem:
\begin{theorem}
    Consider a Whishart matrix $Z$, and let $\mu_{n,m}$ its empirical distribution given by: 
    \begin{equation*}
        \mu_{n,m}:=\frac{1}{n}\sum_{z\in\operatorname{spec}(Z)}\delta_z,
    \end{equation*}
    Assuming that $n/m$ converges to $t$ as $n\to\infty$. Then $\mu_{n,m}$ converges (weakly) to $\MP(t)$ with $t>0$.
\end{theorem}
For a proof and more detailed statement of this result, we refer to Theorem 3.6 and Theorem 3.7 from \cite{bai2010spectral}. In particular, what will an important role in this work is the $\MP$, where the distribution is:
\begin{equation}\label{eq: MP law}
\mathrm    d\MP:=\frac{1}{2\pi}\sqrt{4x^{-1}-1}\,\mathrm dx,
\end{equation}
where we have used the shorthand notation $\MP$ instead of $\MP(1)$.

As for the semicircular distribution described previously, one can relate the moments of free convolution products of $\MP$ to Fuss-Catalan numbers. We shall only give some relevant results for the Mar\u{c}henko-Pastur distribution to be as concise as possible, we refer to \cite{banica2011free} for more details and proofs.
\begin{theorem}
   Let $\MP(t)$ the Mar\u{c}henko-Pastur distribution. The S-transform is given by:
   \begin{equation*}
       S_{\MP(t)}(z)=\frac{1}{t+z}.
   \end{equation*}
    
\end{theorem}
One of the main results of \cite{banica2011free}, relates the free product convolution of $\MP(t)$ and combinatorics, in particular, we shall only give the result for $\MP$ that will be relevant for this work. 
\begin{theorem}\label{th: Mp law and FC}
   Let $\MP$ the Mar\u{c}henko-Pastur distribution. Let $\MP^{\boxtimes s}$ with $s\geq 2$, then we have
   \begin{equation}
       \int_{\R} x^n \mathrm{d}\MP^{\boxtimes s}= \mathrm{FC}_{n,s},
   \end{equation}
   where $\mathrm {FC}_{n,s}$ are the Fuss-Catalan numbers.
\end{theorem}

\vspace{3cm}

\end{document}